\documentclass[12pt]{article}

\usepackage{algorithm}
\usepackage{epstopdf}
\usepackage[noend]{algpseudocode}
\usepackage{amsfonts}
\usepackage{amsmath}
\usepackage{amssymb}
\usepackage{amsthm}
\usepackage{hyperref} 
\usepackage{cleveref}
\usepackage{color}
\usepackage{xcolor}
\usepackage{comment}
\usepackage{dsfont}
\usepackage{enumitem}
\usepackage{fullpage}
\usepackage[utf8]{inputenc}
\usepackage{mathtools}
\usepackage{booktabs}
\usepackage{tabularx}
\usepackage{caption}
\usepackage{subcaption}
\usepackage{natbib}
\usepackage{yfonts} 
\allowdisplaybreaks

\newtheorem{theorem}{Theorem}

\newtheorem{assumption}{Assumption}
\newtheorem{lemma}{Lemma}
\newtheorem{definition}{Definition}
\newtheorem{proposition}{Proposition}

\crefname{assumption}{Assumption}{Assumptions}

\newcommand{\E}{\mathbb{E}}

\DeclareMathOperator*{\argmin}{arg\,min}

\title{Measure Transport with Kernel Stein Discrepancy}
\author{Matthew A. Fisher$^1$, Tui Nolan$^{2,3}$, Matthew M. Graham$^{1,4}$, \\ Dennis Prangle$^1$, Chris. J. Oates$^{1,4}$ \\
\small $^1$Newcastle University, UK \\
\small $^2$Cornell University, US \\
\small $^3$University of Technology Sydney, Australia \\
\small $^4$Alan Turing Institute, UK }

\begin{document}

\maketitle

\begin{abstract}
Measure transport underpins several recent algorithms for posterior approximation in the Bayesian context, wherein a transport map is sought to minimise the Kullback--Leibler divergence (KLD) from the posterior to the approximation. 
The KLD is a strong mode of convergence, requiring absolute continuity of measures and placing restrictions on which transport maps can be permitted.
Here we propose to minimise a kernel Stein discrepancy (KSD) instead, requiring only that the set of transport maps is dense in an $L^2$ sense and demonstrating how this condition can be validated.
The consistency of the associated posterior approximation is established and empirical results suggest that KSD is competitive and more flexible alternative to KLD for measure transport.
\end{abstract}

\section{Introduction}
\label{sec:introduction}

A popular and constructive approach to approximation of complicated distributions is to learn a transformation from a simpler reference distribution.
Within machine learning, neural networks are often used to provide flexible families of transformations which can be optimised by stochastic gradient descent on a suitable objective, with \emph{variational autoencoders} \citep{kingma2013,rezende2014stochastic}, \emph{generative adversarial networks} \citep{goodfellow2014generative}, \emph{generative moment matching networks} \citep{li2015generative,dziugaite2015training} and \emph{normalizing flows}
\citep{rezende2015variational, kingma, dinh2016density, papamakarios2019normalizing, Kobyzev_2020} all fitting in this framework.
The principal application for such generative models is \emph{distribution estimation};
samples are provided from the target distribution and the task is to fit a distribution to these samples. Parallel developments within applied mathematics view the transformation as a \emph{transport map} performing \emph{measure transport} \citep{Marzouk_2016, parno2018transport}. The principal application for measure transport is \emph{posterior approximation}; an un-normalised density function defines the complicated distribution and the task is to approximate it.
In this paper we study posterior approximation, noting that the flexible transformations developed in the machine learning literature can also be applied to this task.

Measure transport provides a powerful computational tool for Bayesian inference in settings that can be challenging for standard approaches, such as Markov chain Monte Carlo (MCMC) or mean field variational inference.
For example, even sophisticated MCMC methods can fail when a posterior is concentrated around a sub-manifold of the parameter space \citep{livingstone2019robustness, au2020manifold}, while it can be relatively straight-forward to define a transport map whose image is the sub-manifold \citep{parno2018transport, brehmer2020flows}.
Likewise, mean field variational inference methods can perform poorly in this context, since independence assumptions can be strongly violated \citep{blei2017variational}.

Let $\mathcal{Y}$ be a measurable space equipped with a probability measure $P$, representing the posterior to be approximated.
The task that we consider in this paper is to elicit a second measurable space $\mathcal{X}$, equipped with a probability measure $Q$, and a measurable function $T : \mathcal{X} \rightarrow \mathcal{Y}$, such that the push-forward $T_\# Q$ (i.e.~the measure produced by applying $T$ to samples from $Q$) approximates $P$, in a sense to be specified. It is further desired that $Q$ should be a ``simple'' distribution that is easily sampled. In contrast to the literature on normalising flows, it is \emph{not} stipulated that $T$ should be a bijection, since we wish to allow for situations where $\mathcal{X}$ and $\mathcal{Y}$ have different cardinalities or where $P$ is supported on a sub-manifold.

A natural starting point is a notion of \emph{discrepancy} $\mathcal{D}(P_1,P_2)$ between two probability measures, $P_1$ and $P_2$, on $\mathcal{Y}$, with the property that $\mathcal{D}(P_1,P_2) = 0$ if and only if $P_1$ and $P_2$ are equal. Then one selects a measurable space $\mathcal{X}$ and associated probability measure $Q$ and seeks a solution to
\begin{equation} \label{eq:minproblem1}
\argmin_{T \in \mathcal{T}} \mathcal{D}(P,T_\# Q),
\end{equation}
over a suitable set $\mathcal{T}$ of measurable functions from $\mathcal{X}$ to $\mathcal{Y}$. A popular choice of $\mathcal{D}$ is the Kullback-Leibler divergence (KLD), giving rise to \emph{variational inference} \citep{blei2017variational}, but other discrepancies can be considered \citep{ranganath2016operator}. The problem in \eqref{eq:minproblem1} can be augmented to include also the selection of $\mathcal{X}$ and $Q$, if desired.

The solution of \eqref{eq:minproblem1} provides an approximation to $P$ whose quality will depend on the set $\mathcal{T}$ and the discrepancy $\mathcal{D}$. This motivates us to consider the choice of $\mathcal{T}$ and $\mathcal{D}$, taking into account considerations that go beyond computational tractability. For example, a desirable property would be that, for a sequence of probability measures $(P_n)_{n\in\mathbb{N}}$, if $\mathcal{D}(P,P_n) \rightarrow 0$ then $P_n \rightarrow P$ in some suitable sense. For $\mathcal{D} = \mathcal{D}_{\text{KL}}$, the KLD\footnote{We use the notation $\mathcal{D}_{\text{KL}}(P,Q) \coloneqq \text{KL}(Q||P)$.}, it holds that $\mathcal{D}_{\text{KL}}(P,P_n) \rightarrow 0$ implies $P_n$ converges to $P$ in \emph{total variation}, from Pinsker's inequality \citep{tsybakov2008nonparametric}. This is a strong mode of convergence, requiring absolute continuity of measures that may be difficult to ensure when the posterior is concentrated near to a sub-manifold. Accordingly, the use of KLD for measure transport places strong and potentially impractical restrictions on which maps $T$ are permitted \citep[e.g.][required that $T$ is a diffeomorphism with $\text{det}\nabla T > 0$ on $\mathcal{X}$]{Marzouk_2016,parno2018transport}. This motivates us in this paper to consider the use of an alternative discrepancy $\mathcal{D}$, corresponding to a weaker mode of convergence, for posterior approximation using measure transport. The advantage of discrepancy measures inducing weaker modes of convergence has also motivated recent developments in generative adversarial networks \citep{arjovsky2017wasserstein}.

Our contributions are as follows:

\begin{itemize} 
	\item We propose kernel Stein discrepancy (KSD) as an alternative to KLD for posterior approximation using measure transport, showing that KSD renders \eqref{eq:minproblem1} tractable for standard stochastic optimisation methods (\Cref{prop: gradients}).
	\item Using properties of KSD we are able to establish consistency under explicit and verifiable assumptions on $P$, $Q$ and $\mathcal{T}$ (\Cref{theorem: convergence theorem}). 
	\item Our theoretical assumptions are weak -- we do not even require $T$ to be a bijection -- and are verified for a particular class of neural network (\Cref{prop: ReLU DNN}).
	In particular, we do not require $Q$ and $P$ to be defined on the same space, allowing quite flexible mappings $T$ to be constructed.
	\item Empirical results support KSD as a competitive alternative to KLD for measure transport.
\end{itemize}

Earlier work on this topic appears limited to \cite{hu2018stein}, who trained a neural network with KSD.
Here we consider general transport maps and we establish consistency of the method, which these earlier authors did not.
We note also that gradient flows provide an alternative (implicit) approach to measure transport \citep{liu2016stein}.

\paragraph{Outline:}
The remainder of the paper is structure as follows: \Cref{sec:methodology} introduces measure transport using KSD, \Cref{sec: theory} contains theoretical analysis for this new method, \Cref{sec:experiments} presents a detailed empirical assessment and \Cref{sec: discuss} contains a discussion of our main findings. 

\section{Methods} \label{sec:methodology}

This section introduces measure transport using KSD.
In \Cref{subsec: measure transport} and \Cref{subsec: RKHS} we recall mathematical definitions from measure transport and Hilbert spaces, respectively; in \Cref{subsec: KSD} we recall the definition and properties of KSD; in \Cref{subsec: methodology} we formally define our proposed method, and in \Cref{subsec: transport maps section} we present some parametric families $\mathcal{T}$ that can be employed.

\paragraph{Notation:}
The set of probability measures on a measurable space $\mathcal{X}$ is denoted $\mathcal{P}(\mathcal{X})$ and a point mass at $x \in \mathcal{X}$ is denoted $\delta(x) \in \mathcal{P}(\mathcal{X})$. For $P \in \mathcal{P}(\mathcal{X})$ let $L^q(P) := \{f : \mathcal{X} \rightarrow \mathbb{R} : \int f^q \mathrm{d}P < \infty\}$. For $P \in \mathcal{P}(\mathbb{R}^d)$ and $(P_n)_{n \in \mathbb{N}} \subset \mathcal{P}(\mathbb{R}^d)$, let $P_n \Rightarrow P$ denote weak convergence of the sequence of measures $(P_n)_{n \in \mathbb{N}}$ to $P$. The Euclidean norm on $\mathbb{R}^n$ is denoted $\|\cdot\|$. Partial derivatives are denoted $\partial_{x}$. For a function $f:\mathbb{R}^n\rightarrow\mathbb{R}$ the gradient is defined as $[\nabla f]_{i} = \partial_{x_i} f$. For a function $f = (f_1,\ldots,f_m) : \mathbb{R}^n \rightarrow \mathbb{R}^m$, the divergence is defined as $\nabla \cdot f = \sum_{i=1}^n \partial_{x_i} f_i$.

Our main results in this paper concern the Euclidean space $\mathbb{R}^d$, but in some parts of the paper, such as \Cref{subsec: measure transport}, it is possible to state definitions at a greater level of generality at no additional effort - in such situations we do so.

\subsection{Measure Transport} \label{subsec: measure transport}

A \emph{Borel space} $\mathcal{X}$ is a topological space equipped with its Borel $\sigma$-algebra, denoted $\Sigma_\mathcal{X}$.
Throughout this paper we restrict attention to Borel spaces $\mathcal{X}$ and $\mathcal{Y}$. Let $Q \in \mathcal{P}(\mathcal{X})$ and $P \in \mathcal{P}(\mathcal{Y})$. In the parlance of measure transport,  $Q$ is the \emph{reference} and $P$ the \emph{target}. Let $T : \mathcal{X} \rightarrow \mathcal{Y}$ be a measurable function and define the \emph{pushforward} of $Q$ through $T$ as the probability measure $T_\# Q \in \mathcal{P}(\mathcal{Y})$ that assigns mass $(T_\# Q)(A) = Q(T^{-1}(A))$ to each $A\in\Sigma_\mathcal{Y}$. Here $T^{-1}(A) = \{x \in \mathcal{X} : T(x) \in A\}$ denotes the pre-image of $A$ under $T$.
Such a function $T$ is called a \emph{transport map} from $Q$ to $P$ if $T_\# Q = P$. 

Faced with a complicated distribution $P$, if one can express $P$ using a transport map $T$ and a distribution $Q$ that can be sampled, then samples from $P$ can be generated by applying $T$ to samples from $Q$. This idea underpins elementary methods for numerical simulation of random variables \citep{devroye2013}. However, in posterior approximation it will not typically be straightforward to identify a transport map and at best one can seek an \emph{approximate} transport map, for which $T_\# Q$ approximates $P$ in some sense to be specified. In this paper we seek approximations in the sense of KSD, which is formally introduced in \Cref{subsec: KSD} and requires concepts in \Cref{subsec: RKHS}, next.

\subsection{Hilbert Spaces} \label{subsec: RKHS}

A Hilbert space $\mathcal{H}$ is a complete inner product space; in this paper we use subscripts, such as $\langle \cdot , \cdot \rangle_{\mathcal{H}}$, to denote the associated inner product. Given two Hilbert spaces $\mathcal{G}$, $\mathcal{H}$, the \emph{Cartesian product} $\mathcal{G} \times \mathcal{H}$ is again a Hilbert space equipped with the inner product $\langle (g_1,h_1), (g_2,h_2) \rangle_{\mathcal{G} \times \mathcal{H}} := \langle g_1, g_2 \rangle_{\mathcal{G}} + \langle h_1, h_2 \rangle_{\mathcal{H}}$.
In what follows we let $\mathcal{B}(\mathcal{H}) := \{ h \in \mathcal{H} : \langle h , h \rangle_{\mathcal{H}} \leq 1 \}$ denote the unit ball in a Hilbert space $\mathcal{H}$.

From the Moore--Aronszajn theorem \citep{aronszajn1950theory}, any symmetric positive definite function $k : \mathcal{Y} \times \mathcal{Y} \rightarrow \mathbb{R}$ defines a unique \emph{reproducing kernel Hilbert space} of real-valued functions on $\mathcal{Y}$, denoted $\mathcal{H}_k$ and with inner-product denoted $\langle \cdot , \cdot \rangle_{\mathcal{H}_k}$. Indeed, $\mathcal{H}_k$ is a Hilbert space characterised by the properties (i) $k(\cdot,y) \in \mathcal{H}_k$ for all $y \in \mathcal{Y}$, (ii) $\langle h , k(\cdot,y) \rangle_{\mathcal{H}_k} = h(y)$ for all $h \in \mathcal{H}_k$, $y \in \mathcal{Y}$. Reproducing kernels are central to KSD, as described next.

\subsection{Kernel Stein Discrepancy} \label{subsec: KSD}

Stein discrepancies were introduced in \cite{gorham2015measuring} to provide a notion of discrepancy that is computable in the Bayesian statistical context. In this paper we focus on so-called \emph{kernel Stein discrepancy} \cite[KSD;][]{liu2016kernelized,chwialkowski2016kernel,gorham2017measuring} since this has lower computational overhead compared to the original proposal of \cite{gorham2015measuring}.

The construction of KSD relies on Stein's method \citep{stein1972} where, for a possibly complicated probability measure $P \in \mathcal{P}(\mathcal{Y})$ of interest, one identifies a \emph{Stein set} $\mathcal{F}$ and a \emph{Stein operator} $\mathcal{A}_P$, such that $\mathcal{A}_P$ acts on elements $f \in \mathcal{F}$ to return functions $\mathcal{A}_P f : \mathcal{Y} \rightarrow \mathbb{R}$ with the property that
\begin{equation}
P' = P \quad \text{iff} \quad \E_{Y \sim P'}[(\mathcal{A}_P f)(Y)] = 0 \; \forall f\in\mathcal{F} \label{eq: stein identity}
\end{equation}
for all $P' \in \mathcal{P}(\mathcal{Y})$.
A \emph{Stein discrepancy} uses the extent to which \eqref{eq: stein identity} is violated to quantify the discrepancy between $P'$ and $P$:
\begin{align*}
\mathcal{D}_{\text{S}}(P,P') := \sup_{f \in \mathcal{F}} |\E_{Y\sim P'}[(\mathcal{A}_P f)(Y)]|
\end{align*}
Note that $\mathcal{D}_{\text{S}}$ is not symmetric in its arguments.
For $\mathcal{Y} = \mathbb{R}^d$ and suitably regular $P$, which admits a positive and differentiable density function $p$, \cite{liu2016kernelized,chwialkowski2016kernel} showed that one may take $\mathcal{F}$ to be a set of smooth vector fields $f : \mathbb{R}^d \rightarrow \mathbb{R}^d$ and $\mathcal{A}_P$ to be a carefully chosen differential operator on $\mathbb{R}^d$. More precisely, and letting $s_p := \nabla \log p$, we have \Cref{theorem:gorham inverse multi quadric} below, which is due to \citet[][Theorem 7]{gorham2017measuring}:

\begin{definition}[\cite{Eberle_2015}]
	A probability measure $P \in \mathcal{P}(\mathbb{R}^d)$ is called \emph{distantly dissipative} if $\liminf_{r\rightarrow \infty} \kappa(r) > 0$, where
	\begin{equation*} 
	\textstyle \kappa(r) \coloneqq -r^{-2} \inf_{\|x-y\| = r} \langle s_p(x) -s_p(y),x-y\rangle. 
	\end{equation*}
\end{definition}

\begin{theorem}
	\label{theorem:gorham inverse multi quadric} Suppose that $P \in \mathcal{P}(\mathbb{R}^d)$ is distantly dissipative. 
	For some $c>0$, $\ell > 0$ and $\beta \in (-1,0)$, let
	\begin{eqnarray}
	\hspace{-5pt} \mathcal{F} := \textstyle \mathcal{B} \big( \prod_{i=1}^d \mathcal{H}_k \big), \quad k(x,y) := (c^2 + \|\frac{x-y}{\ell}\|^2)^\beta \label{eq: IMQ} \\
	\mathcal{A}_P f := f\cdot\nabla \log p + \nabla\cdot f  . \hspace{60pt}
	\end{eqnarray} 
	Then \eqref{eq: stein identity} holds.
	Moreover, if $\mathcal{D}_{\textsc{S}}(P,P_n)\rightarrow 0$, then $P_n \Rightarrow P$. 
\end{theorem}

The kernel $k$ appearing in \eqref{eq: IMQ} is called the \emph{inverse multi-quadric} kernel.
It is known that the elements of $\mathcal{H}_k$ are smooth functions, which justifies the application of the differential operator.
The last part of \Cref{theorem:gorham inverse multi quadric} clarifies why KSD is useful; convergence in KSD controls the standard notion of weak convergence of measures to $P$. 

KSD, in contrast to KLD, is well-defined when the approximating measure $P'$ and the target $P$ differ in their support.
Moreover, in some situations KSD can be exactly computed: from \citet[][Theorem 3.6]{liu2016kernelized} or equivalently \citet[][Theorem 2.1]{chwialkowski2016kernel},
\begin{align}
\mathcal{D}_{\text{S}}(P,P') & = \sqrt{ \E_{Y,Y' \sim P'}[u_p(Y,Y')] } \label{eq:KSD compute} \\
u_p(y,y') & := s_p(y)^\top k(y,y') s_p(y') + s_p(y)^\top \nabla_{y'} k(y,y') \nonumber \\
& \hspace{-5pt} + \nabla_y k(y,y')^\top s_p(y') + \nabla_y\cdot\nabla_{y'} k(y,y'). \label{eq: up formula}
\end{align}
It follows that KSD can be exactly computed whenever $P'$ has a finite support and $s_p$ can be evaluated on this support:
\begin{align} \textstyle
\mathcal{D}_{\text{S}}\left(P, \frac{1}{n} \sum_{i=1}^n \delta(y_i) \right) = \sqrt{\frac{1}{n^2} \sum_{i,j=1}^n u_p(y_i,y_j)} . \label{eq: KSD final}
\end{align}
Computation of \eqref{eq: KSD final} can proceed with $p$ available up to an unknown normalisation constant, facilitating application in the Bayesian context. Now we are in a position to present our proposed method.

\subsection{Measure Transport with KSD} \label{subsec: methodology}

Our proposed method for posterior approximation is simply stated at a high level; we attempt to solve \eqref{eq:minproblem1} with $\mathcal{D} = \mathcal{D}_{\text{S}}$ and over a set $\mathcal{T}$ of candidate functions $T^\theta : \mathcal{X} \rightarrow \mathcal{Y}$ indexed by a finite-dimensional parameter $\theta \in \Theta$. That is, we aim to solve
\begin{equation} \label{eq:minproblem2}
\argmin_{\theta \in \Theta} \mathcal{D}_{\text{S}}(P,T_\#^\theta Q).
\end{equation}
Discussion of the choice of $\mathcal{T}$ is deferred until \Cref{subsec: transport maps section}.
Compared to previous approaches to measure transport using KLD \citep{rezende2015variational,kingma,Marzouk_2016,parno2018transport}, KSD is arguably more computationally and theoretically tractable; the computational aspects will now be described.

The solution of \eqref{eq:minproblem2} is equivalent to minimisation of the function $F(\theta) := \mathcal{D}_{\text{S}}(P,T_\#^\theta Q)^2$ over $\theta \in \Theta$. In order to employ state-of-the-art algorithms for stochastic optimisation, an unbiased estimator for the gradient $\nabla_\theta F(\theta)$ is required. A naive starting point would be to differentiate the expression for the KSD of an empirical measure in \eqref{eq: KSD final}, however the resulting \emph{V-statistic} is biased. Under weak conditions, we establish instead the following unbiased estimator (a \emph{U-statistic}) for the gradient:

\begin{proposition} \label{prop: gradients}
	Let $\Theta \subseteq \mathbb{R}^p$ be an open set.
	Assume that $\forall \theta \in \Theta$
	\begin{enumerate}[leftmargin=25pt]
		\item[(A1)] $(x,x') \mapsto u_p(T^\theta(x),T^\theta(x'))$ is measurable;
		\item[(A2)] $\mathbb{E}_{X,X' \sim Q}\left[ |u_p(T^\theta(X),T^\theta(X'))|\right] < \infty$;
		\item[(A3)] $\mathbb{E}_{X,X' \sim Q}\left[ \|\nabla_{\theta} u_p(T^\theta(X),T^\theta(X'))\|\right] < \infty$;
	\end{enumerate}
	and that $\forall x,x' \in \mathcal{X}$, 
	\begin{enumerate}[leftmargin=25pt]
		\item[(A4)] $\theta \mapsto \nabla_\theta u_p(T^\theta(x),T^\theta(x'))$ is continuous.
	\end{enumerate}
	Then $\forall \theta \in \Theta$
	\begin{align*} 
	\nabla_\theta F(\theta)
	= \textstyle \E\Big[ \frac{1}{n(n-1)} \sum\limits_{i \neq j} \nabla_\theta u_p(T^\theta(x_i),T^\theta(x_j)) \Big] ,
	\end{align*}
	where the expectation is taken with respect to independent samples $x_1,\dots,x_n \sim Q$.
\end{proposition}
All proofs are contained in \Cref{sec: proofs}.
The assumptions on $u_p$ amount to assumptions on $T$, $p$ and $k$, by virtue of \eqref{eq: up formula}.
It is not difficult to find explicit assumptions on $T$, $p$ and $k$ that imply (A1-4), but these may be stronger than required and we prefer to present the most general result.

Armed with an unbiased estimator of the gradient, we can employ a stochastic optimisation approach, such as stochastic gradient descent \citep[SGD;][]{robbins1951} or adaptive moment estimation \citep[Adam;][]{kingma2014adam}. 
See \cite{kushner2003stochastic, ruder2016overview}.
For the results reported in the main text we used Adam, with $\theta$ initialised as described in \Cref{subsec: initialise}, but other choices were investigated (see \Cref{subsec: stochastic optimisation}).

\subsection{Parametric Transport Maps} \label{subsec: transport maps section}

In this section we describe some existing classes of transport map $T : \mathcal{X} \rightarrow \mathcal{Y}$ that are compatible with KSD measure transport. 
From \Cref{prop: gradients} we see that measure transport using KSD does not impose strong assumptions on the transport map.
Indeed, compared to KLD \citep{rezende2015variational,kingma,Marzouk_2016,parno2018transport} we do not require that $T$ is a diffeomorphism ($T$ need not even be continuous, nor a bijection), making our framework considerably more general.
This additional flexibility may allow measure to be transported more efficiently, using simpler maps.
That being said, if one wishes to compute the density of $T_\# Q$ (in addition to sampling from $T_\# Q$), then a diffeomorphism, along with the usual change-of-variables formula, should be used.

\paragraph{Triangular Maps:}
\cite{rosenblatt1952remarks} and \cite{knothe1957contributions} observed that, for $P,Q \in \mathcal{P}(\mathbb{R}^d)$ admitting densities, a transport map $T:\mathbb{R}^d\rightarrow\mathbb{R}^d$ can without loss of generality be sought in the \emph{triangular form}
\begin{equation} \label{eq:lowertriangular}
T(x) = (T_1(x_1), T_2(x_1,x_2), \dots, T_d(x_1,\ldots,x_d) ),
\end{equation}
where each $T_i:\mathbb{R}^i \rightarrow \mathbb{R}$ and $x = (x_1,\ldots,x_d)$ \citep[][Lemma 2.1]{bogachev2005triangular}. The triangular form was used in \cite{Marzouk_2016,parno2018transport}, since the Jacobian determinant, that is required when using KLD (but not KSD), can exploit the fact that $\nabla T$ is triangular to maintain linear complexity in $d$. 

\paragraph{Maps from Measure Transport:}
In the context of a triangular map $T = (T_1,\ldots,T_d)$, \cite{Marzouk_2016} and \cite{parno2018transport} considered several parametric models for the components $T_i$, including polynomials, radial basis functions and monotone parameterisations of the form
\begin{align*}
T_i(x_1,\ldots,x_i) & = f_i(x_1,\ldots,x_{i-1}) \\
& \qquad + \textstyle \int_0^{x_i} \exp(g_i(x_1,\ldots,x_{i-1},y))\,\mathrm{d}y,
\end{align*}
for functions $f_i:\mathbb{R}^{i-1}\rightarrow\mathbb{R}$ and $g_i:\mathbb{R}^i\rightarrow\mathbb{R}$. 
The monotone parameterisation ensures that $\text{det}\nabla T > 0$ on $\mathbb{R}^d$, which facilitates computation of the density of $T_\#P$, as required for KLD\footnote{For polynomials and radial basis functions, these authors only enforced $\text{det}\nabla T > 0$ locally, introducing an additional approximation error in evaluation of KLD; such issues do not arise with KSD.}.

\paragraph{Maps from Normalising Flows:}
The principal application of normalising flows is density estimation \citep{papamakarios2019normalizing,Kobyzev_2020}, but the parametric families of transport map used in this literature can also be used for posterior approximation \citep{rezende2015variational}. A normalising flow is required to be a diffeomorphism $T : \mathbb{R}^d \rightarrow \mathbb{R}^d$ with the property that the density of $T_\#Q$ can be computed. A popular choice that exploits the triangular form \eqref{eq:lowertriangular} is an \textit{autoregressive flow} $T_i(x) = \tau(c_i(x_1,\ldots,x_{i-1}),x_i)$, where $\tau$ is a monotonic transformation of $x_i$ parameterised by $c_i$, e.g.~an affine transformation $T_i(x) = \alpha_i x_i + \beta_i$ where $c_i$ outputs $\alpha_i \neq 0$ and $\beta_i$. For instance, \cite{kingma} proposed \textit{inverse autoregressive flows} (IAF), taking $T(x) = \mu + \exp(\sigma) \odot x$.
Here $\odot$ is elementwise multiplication and $\mu$ and $\sigma$ are vectors output by an autoregressive neural network: one designed so that $\mu_i, \sigma_i$ depend on $x$ only through $x_j$ for $j<i$. In \cite{huang2018neural}, $\tau$ was the output of a monotonic neural network and the resulting flow was called a \textit{neural autoregressive flow} (NAF). Compositions of normalising flows can also be considered, of the form
\begin{equation}
T = T^{(n)} \circ \dots \circ T^{(1)}
\label{eq: NF compose}
\end{equation}
where each $T^{(i)}$ is itself a normalising flow e.g.~a IAF. For instance, \cite{dinh2014nice} proposed using \textit{coupling layers} of the form $T^{(i)}(x) = (h(x_1,\ldots,x_r),x_{r+1},\ldots,x_d)$, where $r < d$ and $h:\mathbb{R}^r\rightarrow\mathbb{R}^r$ is a bijection. These only update the first $r$ components of $x$, so they are typically composed with permutations.

Regardless of the provenance of a transport map $T$, all free parameters of $T$ are collectively denoted $\theta$, and are to be estimated. The suitability of a parametric set of candidate maps in combination with KSD is studied both empirically in \Cref{sec:experiments} and theoretically, next.

\section{Theoretical Assessment} \label{sec: theory}

In \Cref{subsec: exist} we affirm basic conditions on $P$ and $Q$ for a transport map to exist. In \Cref{subsec: consistent} we establish sufficient conditions for the consistency of our method and in \Cref{subsec: validate} we consider a particular class of transport maps based on neural networks, to demonstrate how our conditions on the transport map can be explicitly validated.

\subsection{Existence of an $L^2$ Transport Map} \label{subsec: exist}

For a complete separable metric space $\mathcal{X}$, recall that the \emph{Wasserstein space} of order $p \geq 1$ is defined by taking some $x_0 \in \mathcal{X}$ and
\begin{equation*}
\mathcal{P}_p(\mathcal{X}) \coloneqq \textstyle \left\{P\in\mathcal{P}(\mathcal{X}) :  \int \text{dist}(x,x_0)^p\,\mathrm{d}P(x) < \infty \right\} ,
\end{equation*}
where the definition is in fact independent of the choice of $x_0 \in \mathcal{X}$ \citep[][Definition 6.4]{villani_2009}.
For existence of a transport map, we make the following assumptions on $P$ and $Q$:

\begin{assumption}[Assumptions on $Q$] \label{assumption reference}
	The reference measure $Q\in\mathcal{P}(\mathcal{X})$, where $\mathcal{X}$ is a complete separable metric space, and $Q(\{x\}) = 0$ for all $x\in\mathcal{X}$.
\end{assumption}

\begin{assumption}[Assumptions on $P$] \label{assumption target}
	The target measure $P\in\mathcal{P}_2(\mathbb{R}^d)$ has a strictly positive density $p$ on $\mathbb{R}^d$. 
\end{assumption}

These assumptions guarantee the existence of a transport map with $L^2$ regularity, as shown in the following result:

\begin{proposition} \label{prop: transport map properties}
	If \Cref{assumption reference,assumption target} hold, then there exists a transport map $T \in \prod_{i=1}^d L^2(Q)$ such that $T_\# Q = P$.
\end{proposition}

Of course, such a transport map will not be unique in general.

\subsection{Consistent Posterior Approximation} \label{subsec: consistent}

The setting for our theoretical analysis considers a sequence $(\mathcal{T}_n)_{n \in \mathbb{N}}$ of parametric classes of transport map, where intuitively $\mathcal{T}_n$ provides a more flexible class of map as $n$ is increased.
For example, $\mathcal{T}_n$ could represent the class of triangular maps comprising of $n$th order polynomials, or a class of normalising flows comprising of $n$ layers in \eqref{eq: NF compose}.

\begin{assumption}[Assumptions on $\mathcal{T}_n$] \label{assum: G contains T}
	There exists a subset $\textgoth{T} \subseteq \prod_{i=1}^d L^2(Q)$ containing an element $T\in\textgoth{T}$ for which $T_\# Q=P$.
	The sequence $(\mathcal{T}_n)_{n \in \mathbb{N}}$ satisfies $\mathcal{T}_n \subseteq \textgoth{T}$ with $\mathcal{T}_n\subseteq\mathcal{T}_m$ for $n \leq m$ and $\mathcal{T}_\infty := \lim_{n \rightarrow \infty} \mathcal{T}_n$ is a dense set in $\textgoth{T}$.
\end{assumption} 

\Cref{prop: transport map properties} provides sufficient conditions for the set $\textgoth{T}$ in \Cref{assum: G contains T} to exist; the additional content of \Cref{assum: G contains T} ensures that $\mathcal{T}_\infty$ is rich enough to consistently approximate an exact transport map, in principle at least.
Next, we state our consistency result:

\begin{theorem}\label{theorem: convergence theorem}
	Let \Cref{assumption reference,assumption target,assum: G contains T} hold.
	Further suppose that $P$ is distantly dissipative, with $\nabla \log p$ Lipschitz and $\E_{X\sim P}[\|\nabla \log p(X)\|^2] < \infty$. 
	Suppose that $T_n \in \mathcal{T}_n$ satisfies
	\begin{equation}
	\mathcal{D}_{\textsc{S}}(P,(T_n)_\#Q) - \inf_{T\in\mathcal{T}_n} \mathcal{D}_{\textsc{S}}(P,T_\#Q) \stackrel{n\rightarrow \infty}{\rightarrow} 0 , \label{eq: main result}
	\end{equation}
	with $\mathcal{D}_{\textsc{S}}$ defined  in \Cref{theorem:gorham inverse multi quadric}.
	Then $(T_n)_\#Q \Rightarrow P$.
\end{theorem}

The statement in \eqref{eq: main result} accommodates the reality that, although finding the global optimum $T \in \mathcal{T}_n$ will typically be impractical, one can realistically expect to find an element $T_n$ that achieves an almost-as-low value of KSD, e.g. using a stochastic optimisation method.
To our knowledge, no comparable consistency guarantees exist for measure transport using KLD.

\subsection{Validating our Assumptions on $\mathcal{T}_n$} \label{subsec: validate}

Recall that earlier work on measure transport placed strong restrictions on the set of maps $\mathcal{T}_n$, requiring each map to be a diffeomorphism with non-vanishing Jacobian determinant.
In contrast, our assumptions on $\mathcal{T}_n$ are almost trivial; we do not require smoothness and there is not a bijection requirement.
Our assumptions can be satisfied \emph{in principle} whenever $\mathcal{X}$ is a complete separable metric space, since then $\prod_{i=1}^d L^2(Q)$ is separable \citep[][Proposition 3.4.5]{cohn2013measure} and admits a Schauder basis $\{\phi_i\}_{i \in \mathbb{N}}$, so we may take $\mathcal{T}_n = \text{span}\{\phi_1,\dots,\phi_n\}$ for \Cref{assum: G contains T} to hold. 
In practice we are able to verify \Cref{assum: G contains T} for quite non-trivial classes of map $\mathcal{T}_n$.
To demonstrate, one such example is presented next:

We consider deep neural networks with multi-layer perceptron architecture and ReLU activation functions
Let $\mathcal{R}_{l,n}(\mathbb{R}^p\rightarrow\mathbb{R}^d)$ denote the set of such \emph{ReLU neural networks} $f : \mathbb{R}^p \rightarrow \mathbb{R}^d$ with $l$ layers and width at most $n$. See \Cref{def: DNN} in \Cref{app: DNN} for a formal definition.

\begin{proposition} \label{prop: ReLU DNN}
	Let \Cref{assumption reference,assumption target} hold.
	Let $Q$ admit a positive, continuous and bounded density on $\mathcal{X} = \mathbb{R}^p$.
	Let $\mathcal{T}_n = \mathcal{R}_{l,n}(\mathbb{R}^p\rightarrow\mathbb{R}^d)$ with $l := \lceil \log_2(p+1) \rceil$. 
	Then \Cref{assum: G contains T} holds.
\end{proposition}

The maps in \Cref{prop: ReLU DNN} are not bijections, illustrating the greater flexibility of KSD compared to KLD for measure transport.
This completes our theoretical discussion, and our attention now turns to empirical assessment.

\section{Empirical Assessment} \label{sec:experiments}

The purpose of this section is to investigate whether KSD is competitive with KLD for measure transport. 
\Cref{subsec: toy examples} compares both approaches using a variety of transport maps and a synthetic test-bed. 
Then, in \Cref{subsec: biochemical oxygen,subsec: lotka-volterra} we consider more realistic posterior approximation problems arising from, respectively, a biochemical oxygen model and a parametric differential equation model.

In all experiments we used the kernel \eqref{eq: IMQ} with $c=1$, $\ell = 0.1$, $\beta = -1/2$ (other choices were investigated in \Cref{subsec: misspecified lengthscale}), the stochastic optimiser Adam with batch size $n=100$ and learning rate $0.001$ (other choices were investigated in \Cref{subsec: stochastic optimisation}), and the reference distribution $Q$ was taken to be a standard Gaussian on $\mathbb{R}^p$ (other choices were considered in \Cref{subsec: other Q}).
Code to reproduce these results is available at \url{https://github.com/MatthewAlexanderFisher/MTKSD}.

\subsection{Synthetic Test-Bed} \label{subsec: toy examples}

\begin{figure*}[t!] 
	\centering
	\includegraphics[width=1.\textwidth]{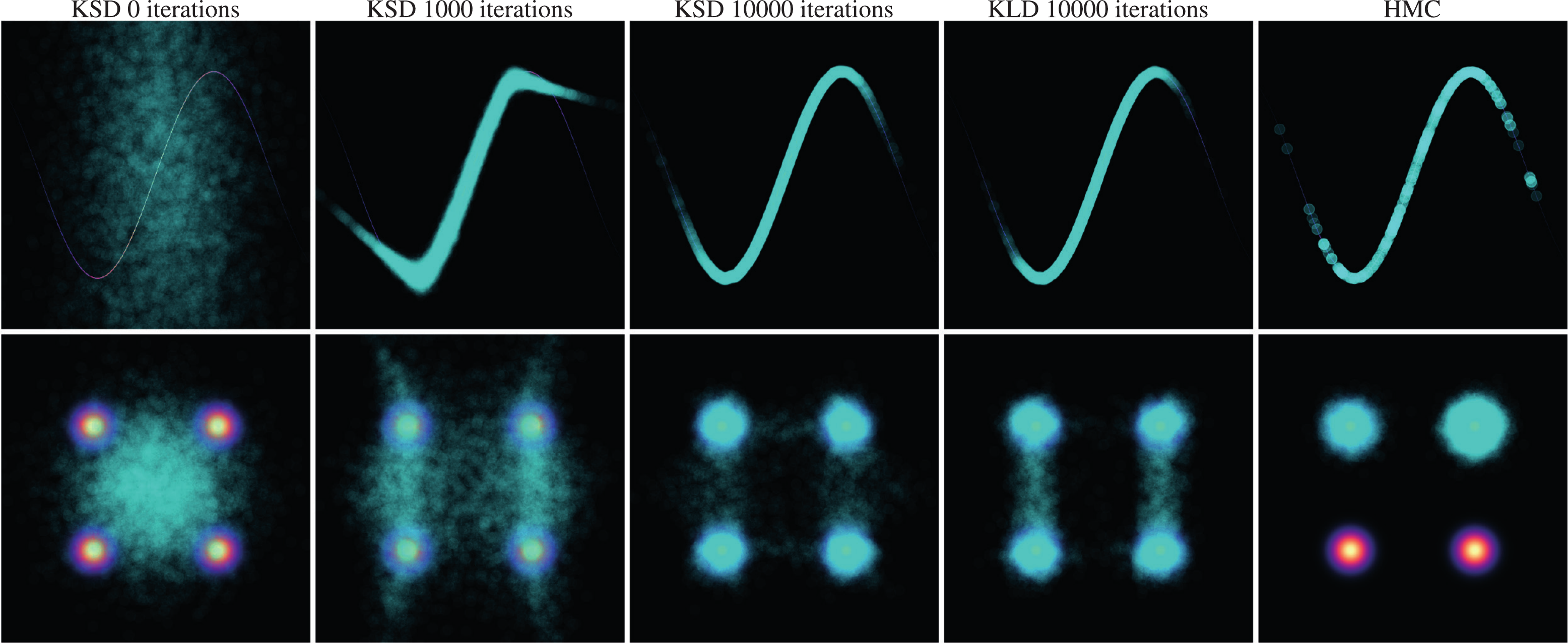}
	
	\caption{
		Measure transport with KSD, versus KLD and HMC.
		The top row reports results for approximation of a sinusoidal density using an inverse autoregressive flow, while the bottom row reports analogous results for a multimodal density and a neural autoregressive flow. 
		The first three columns display convergence of the KSD-based method as the number of iterations of stochastic optimisation is increased.
		The remaining columns compare the output of the KLD-based method and HMC for an identical computational budget.
	}
	\label{fig: toy densities}
\end{figure*}

First we consider a set of synthetic examples that have previously been used to motivate measure transport as an alternative to MCMC.
Three targets were considered; $p_1$ is a sinusoidal density, $p_2$ is a banana density and $p_3$ is multimodal; these are formally defined in \Cref{subsec: details of toy models}. 
Results for $p_1$ and $p_3$ are displayed in \Cref{fig: toy densities}. 
The convergence of the approximation to the target is shown for KSD and the corresponding approximation after $10^4$ iterations of Adam is shown for KLD.
Since, for both objectives, one iteration requires $10^2$ evaluations of $\log p_i$ or its gradient, this represents a total of $10^6$ calls to $\log p_i$ or its gradient.
The corresponding approximation produced using an adaptive Hamiltonian Monte Carlo (HMC) algorithm \citep{hoffman2014no,betancourt2017conceptual} is shown, where the HMC chains were terminated once $10^6$ evaluations of $\log p_i$ or its gradient had been performed. 
Both $p_1$ and $p_3$ present challenges for HMC that, to some extent, can be overcome using measure transport.

The results in \Cref{fig: toy densities} are for a fixed class of transport map, but now we report a systematic comparison of KSD and KLD.
The majority of maps that we consider are diffeomorphic (in order that KLD can be used), implemented in Pyro \citep{bingham2018pyro}.
Since KSD does not place such requirements on the transport map, we also report results for a (non-bijective) ReLU neural network.
Our performance measure is an estimate of the Wasserstein-1 distance between the target and approximate distributions computed using $10^4$ samples (see \Cref{subsec: performance metrics} for details). Results are detailed in \Cref{table: toy density results}. Overall, there is no clear sense in which KSD out-performs KLD or \textit{vice versa}; KSD performed best on $p_1$, KLD performed best on $p_2$, and for $p_3$ the results were mixed. We conclude that these objectives offer similar performance for measure transport. However, KLD cannot be applied to the ReLU neural network (denoted N/A in \Cref{table: toy density results}) due to the strong constraints on the mapping that are required by KLD. 
\begin{table*}[t!] 
	\newcolumntype{Y}{>{\raggedright\arraybackslash}X}

\begin{tabularx}{\textwidth}{@{}llYYYYYY@{}} \toprule
        &  & \multicolumn{2}{c}{Sinusoidal}  & \multicolumn{2}{c}{Banana} & \multicolumn{2}{c}{Multimodal} \\
      \cmidrule(lr){3-4} \cmidrule(l){5-6} \cmidrule(l){7-8} 
     {Transport Map} & {$N$} & {KSD} & {KLD} & {KSD} & {KLD} &  {KSD} & {KLD} \\ \midrule
     IAF  & $10^4$                & $\textbf{0.38}$  & $0.52$ & 0.20 & \textbf{0.07}  & \textbf{0.67} & 1.1 \\
     IAF (stable) & $10^4$        & $\textbf{0.35}$ & $0.39$ & 0.16  & \textbf{0.11} & \textbf{0.61}  & 0.62 \\ 
     NAF   & $10^4$               & $\textbf{0.55}$ & $0.64$ & 0.39  & \textbf{0.025} &  \textbf{0.095} & 0.11 \\ 
     SAF  & $10^4$                & $\textbf{0.23}$ & $0.58$ & 0.20  & \textbf{0.18} & \textbf{0.30}  & 0.48 \\ 
     B-NAF  & $10^4$                & $\textbf{0.78}$ & $1.2$ & 0.70  & \textbf{0.18} & 1.0   & \textbf{0.99}  \\ 
     Polynomial (cubic)  & $10^4$ & $\textbf{0.40}$ & $0.84$ & 0.25  & \textbf{0.059} & 0.51  & \textbf{0.43} \\ 
     IAF mixture  & $3{\times}10^4$        & $1.29$ & $\textbf{0.61}$ & 0.19  & \textbf{0.14} & 0.037  & \textbf{0.036} \\ 
     ReLU network & $5{\times} 10^4$        & $\textbf{0.71}$ & N/A & \textbf{0.43}  & N/A & \textbf{0.22}  & N/A \\\bottomrule
\end{tabularx}

	\caption{Results from the synthetic test-bed. 
		The first column indicates which parametric class of transport map was used; full details for each class can be found in \Cref{subsec: details of toy models}.
		A map-dependent number of iterations of stochastic optimisation, $N$, are reported - this is to ensure that all optimisers approximately converged. 
		The main table reports the (first) Wasserstein distance between the approximation $T_\#Q$ and the target $P$. 
		Bold values indicate which of KSD or KLD performed best.
	}
	\label{table: toy density results}
\end{table*}

Two discussion points are now highlighted:
First, it is known that certain normalising flows can capture multiple modes due to their flexibility, however others cannot \citep{huang2018neural}. One solution is to consider a mixture of transport maps; i.e. $\sum_{i=1}^dw_i T^{(i)}_\# Q_i$ with reference distribution $Q_1 \times \dots \times Q_d$ and mixing weights $w_i>0$ satisfying $\sum_i w_i = 1$. This idea has been explored recently in \cite{pires2020variational}. In \Cref{table: toy density results} we report results using mixtures of \textit{inverse autoregressive flows} (IAF). As one might hope, these approximations were successful in finding each of the modes in $p_3$, but fared relatively worse for $p_1$ and $p_2$. Second, since in Adam we are using a Monte Carlo estimator of the gradient, it is natural to ask whether a quasi Monte Carlo estimator would offer an improvement \citep{wenzel2018quasi}. This was investigated and our results are reported in \Cref{subsec: QMC investigation}.

\subsection{Biochemical Oxygen Demand Model} \label{subsec: biochemical oxygen}

Next we reproduce an experiment that was used to illustrate measure transport using KLD in \cite{parno2018transport}. 
The task is parameter inference in a $d=2$ dimensional oxygen demand model, of the form $B(t) = \alpha_1(1 - \exp( - \alpha_2 t))$, where $B(t)$ is the biochemical oxygen demand at time $t$, a measure of the consumption of oxygen in a given water column sample due to the decay of organic matter \citep{sullivan2010biochemical}. The parameters to be inferred are $\alpha_1 , \alpha_2 > 0$. Full details of the prior and the likelihood are contained in \Cref{subsec: details of biochemical model}. 

For our experiment, we trained a \textit{block neural autoregressive flow}\footnote{This class of transport map was experimentally observed to outperform the other classes we considered.} using $N = 30,000$ iterations of Adam. Results are presented in \Cref{fig: bod results}. Unlike the synthetic experiments, we no longer have a closed form for the target $P$; however, this problem was amenable to MCMC and a long run of HMC ($10^6$ iterations, thinned by a factor of 100) provided a gold standard, allowing us to approximate the Wasserstein-1 distance from $T_\#Q$ to $P$ as in \Cref{subsec: toy examples}. For the KSD-based method, we obtained a Wasserstein-1 distance of $0.069$, while KLD achieved $0.015$. Although the Wasserstein-1 distance for KSD is larger than that for KLD, both values are close to the \textit{noise floor} for our approximation of the Wasserstein-1 distance; two independent runs of HMC ($10^6$ iterations, thinned by a factor of 100), differed in Wasserstein-1 distance by 0.022.
We therefore conclude that KSD and KLD performed comparably on this task.

\begin{figure*}[h!] 
	\centering
	\includegraphics[width=0.9\textwidth]{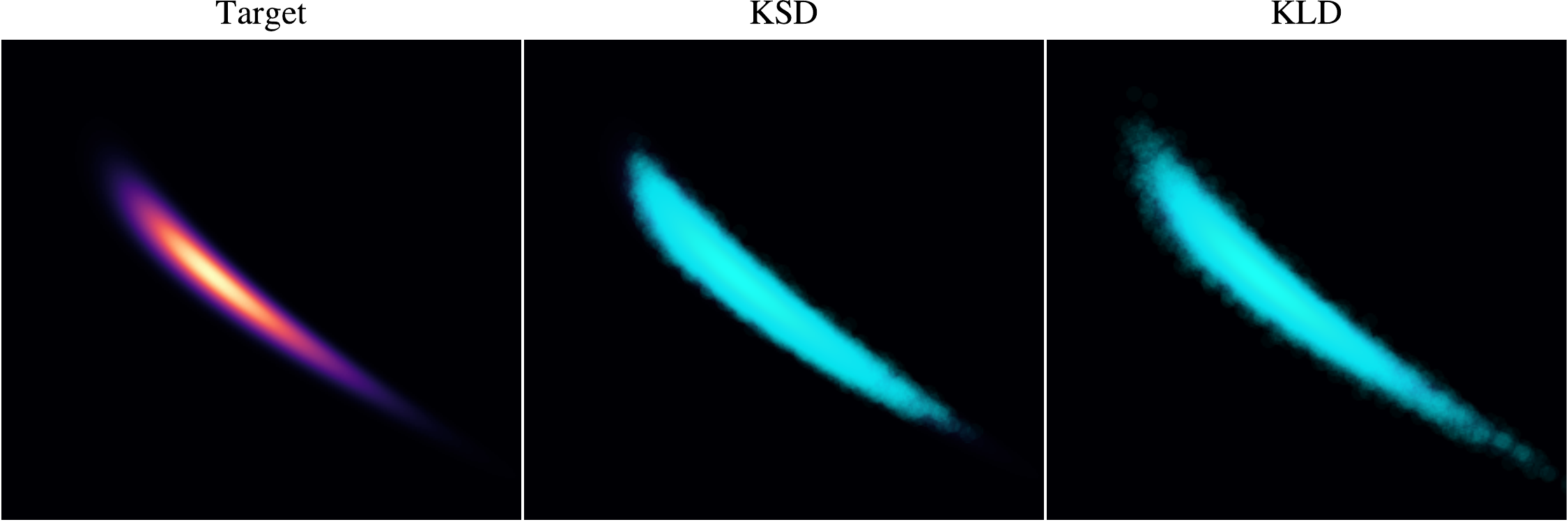} 
	\caption{
		Results for the biochemical oxygen demand model. 
		The leftmost panel is the target distribution, while the central and rightmost panels show samples generated from the output of the methods based, respectively, on KSD and KLD.
	}
	\label{fig: bod results}
\end{figure*}

\subsection{Generalised Lotka-Volterra Model} \label{subsec: lotka-volterra}

Our final experiment is a realistic inference problems involving a non-trivial likelihood. 
Following \cite{parno2018transport}, we consider parameter inference for a generalised Lotka--Volterra model
\begin{align}
\label{eq: LV}
\begin{split}
\textstyle \frac{\mathrm{d}p}{\mathrm{d}t}(t) &= \textstyle rp(t)\big(1 - \frac{p(t)}{k}\big) - s\frac{p(t)q(t)}{a + p(t)} , \\
\textstyle \frac{\mathrm{d}q}{\mathrm{d}t}(t) &= \textstyle u\frac{p(t)q(t)}{a + p(t)} - vq(t), 
\end{split}
\end{align}
where $p(t),q(t) > 0$ are the predator and prey populations respectively at time $t$ and $r, k, s, u, a$ and $v$, along with the initial conditions $p(0) = p_0$ and $q(0) = q_0$, are parameters to be inferred. 
Together, these $d=8$ parameters were inferred from a noisy dataset, with the prior and likelihood reported in \Cref{subsec: details of lotka volterra}.
This task is realistic and yet amenable to MCMC; the latter is an essential requirement to allow us to provide a gold standard against which to assess KSD and KLD, and we again used an extended run of HMC.

For this experiment, the B-NAF class and $N = 5 \cdot 10^4$ iterations of Adam were used.
The gradients, required both for HMC and KSD measure transport, were computed using automatic differentiation through the numerical integrator used to solve \eqref{eq: LV}, implemented in the \texttt{torchdiffeq} Python package \citep{chen2018neural}.

For the KSD-based method, we obtained an approximate Wasserstein-1 distance from $T_\#Q$ to $P$ of $0.130$, while KLD achieved $0.110$.
The noise floor for our approximation of the Wasserstein-1 distance in this case was $0.107$. 
We therefore conclude that KSD and KLD also performed comparably on this more challenging task.

\section{Discussion} \label{sec: discuss}

This paper proposed and studied measure transport using KSD, which can be seen as an instance of \emph{operator variational inference} \citep{ranganath2016operator}.
Our findings suggest that KSD is a suitable variational objective for measure transport; we observed empirical performance comparable with that of KLD, yet only minimal and verifiable conditions on the map $T$ were required.

There are three potential limitations of KSD compared to KLD: First, the parameters of the kernel must be specified, and a poor choice of kernel parameters can result in poor approximation; see \Cref{subsec: misspecified lengthscale}. It would be interesting to explore whether adversarial maximisation of KSD with respect to the kernel parameters, while minimising KSD over the choice of transport map, offers a solution \citep{grathwohl2020learning}. Second, while only first order derivatives are required for KLD, gradient-based optimisation of KSD requires second order derivatives of $p$. In most automatic differentiation frameworks, and for most models, this is possible at little extra computational cost, but sometimes this will present difficulties e.g. for models with differential equations involved. Third, it is known that score-based variational objectives can sometimes exhibit pathologies \citep{wenliang2020blindness}; some of these are illustrated in \Cref{subsec: pathologies}.

Several recent works explored the possibility of combining measure transport with Monte Carlo \citep{salimans2015markov,wolf2016variational,hoffman2017learning,caterini2018hamiltonian,prangle2019distilling,thin2020metflow} and it would also be interesting to consider the use of KSD in that context. Related, for both KSD and KLD there is freedom to select the space $\mathcal{X}$ and the reference distribution $Q$. This could also be handled within the optimisation framework, but further work would be needed to determine how these additional degrees of freedom should be parametrised.

\vspace{5pt}
\noindent
\textbf{Acknowledgements:}
MAF was supported by the EPSRC Centre for Doctoral Training in Cloud Computing for Big Data EP/L015358/1 at Newcastle University, UK.
THN was supported by a Fulbright scholarship, an American Australian Association scholarship and a Roberta Sykes scholarship.
MMG and CJO were supported by the Lloyd's Register Foundation programme on data-centric engineering at the Alan Turing Institute, UK.
The authors thank Onur Teymur for helpful comments on the manuscript.

\bibliographystyle{plainnat}
\bibliography{bibliography}

\newpage
\appendix
\section*{Supplement}

This supplement is structured as follows:

\begin{itemize}
    \item \Cref{sec: proofs} contains the proofs of the theory developed in the main text: \Cref{prop: gradients} (\Cref{app: grad proof}), \Cref{prop: transport map properties} (\Cref{app: T prop}), \Cref{theorem: convergence theorem} (\Cref{app: thm pf}) and \Cref{prop: ReLU DNN} (\Cref{app: DNN}).
    \item \Cref{sec: computational details} contains the computational details and extensions to the experiments detailed in \Cref{sec:experiments}. \Cref{subsec: performance metrics} explains how the Wasserstein-1 distance was computed as our performance metric. \Cref{subsec: details of toy models} provides the full details of the synthetic test-bed experiments used in \Cref{subsec: toy examples}. Similarly, \Cref{subsec: details of biochemical model} and \Cref{subsec: details of lotka volterra} provides the full details of the biochemical oxygen demand experiment of \Cref{subsec: biochemical oxygen} and \Cref{subsec: lotka-volterra} respectively.
    \item \Cref{sec: investigations} contains further investigations into the methods presented in the main paper. \Cref{subsec: initialise} discusses the sensitivity to initialisation and explains how the transport maps were initialised for the experiments of \Cref{sec:experiments}.  \Cref{subsec: stochastic optimisation} explores variations on the stochastic optimisation method. \Cref{subsec: QMC investigation} explores the use of quasi Monte Carlo in the stochastic approximation of gradients for KSD. In \Cref{subsec: other Q}, we investigate the effect of the changing the reference distribution. In \Cref{subsec: misspecified lengthscale}, we explore the effect of the length-scale parameter $\ell$ on KSD-based measure transport. \Cref{subsec: relu investigation} explores the effect of the input dimension in the ReLU network transport map. In \Cref{subsec: U vs V}, we investigate KSD-based measure transport using a (biased) V-statistic estimator of KSD against the unbiased U-statistic estimator that was used for the experiments presented in the main text. Finally, in \Cref{subsec: pathologies}, we document certain pathological behaviours experienced when using KSD for measure transport and offer potential remedies.
\end{itemize}

\section{Proof of Theoretical Results} \label{sec: proofs}

This section contains the proofs for all novel theoretical results stated in the main text.
In \Cref{app: grad proof} we present the proof of \Cref{prop: gradients}; in \Cref{app: T prop} we present the proof of \Cref{prop: transport map properties}; the proof of \Cref{theorem: convergence theorem} is contained in \Cref{app: thm pf}, and finally, the proof of \Cref{prop: ReLU DNN} is in \Cref{app: DNN}.

\subsection{Proof of \Cref{prop: gradients}} \label{app: grad proof}

The argument involved in the proof of \Cref{prop: gradients} requires differentiation under an integral.
The measure-theoretic calculus result that we exploit to justify the interchange of differentiation and integration (\Cref{lem: interchange} below) requires the following mathematical concepts:

\begin{definition}
Let $\Omega$ be a measurable space and let $\Theta$ be a topological space.
A function $f : \Theta \times \Omega \rightarrow \mathbb{R}$ is a \emph{Carath\'{e}odory function} if for each $\theta \in \Theta$ the map $\omega \mapsto f(\theta,\omega)$ is measurable and for each $\omega \in \Omega$ the map $\theta \mapsto f(\theta,\omega)$ is continuous.
\end{definition}

\begin{definition} \label{def: LUIB}
Let $\Omega$ be a measurable space equipped with a measure $\mu$ and let $\Theta$ be a topological space.
A function $f : \Theta \times \Omega \rightarrow \mathbb{R}$ is \emph{locally uniformly integrably bounded} if for every $\theta \in \Theta$ there is a non-negative measurable function $h_\theta : \Omega \rightarrow \mathbb{R}$ such that $\int_\Omega h_\theta(\omega) \mathrm{d}\mu(\omega) < \infty$, and there exists a neighbourhood $U_\theta$ of $\theta$ such that for all $\vartheta \in U_\theta$ we have $|f(\vartheta,\omega)| \leq h_\theta(\omega)$.
\end{definition}

The following sufficient condition for a function to be locally uniformly integrably bounded will be used:

\begin{lemma} \label{lem: LUIB}
In the setting of \Cref{def: LUIB}, let $\Theta \subseteq \mathbb{R}^d$ be an open set and assume further that, for each $\omega \in \Omega$, the function $\theta \mapsto f(\theta,\omega)$ is continuous and that, for each $\theta \in \Theta$, the integral $\int_\Omega |f(\theta,\omega)| \mathrm{d}\mu(\omega) < \infty$ exists.
Then $f$ is locally uniformly integrably bounded.
\end{lemma}
\begin{proof}
Fix $\theta \in \Theta$ and $\omega \in \Omega$.
Since $f(\theta,\omega)$ is continuous in $\theta$ and $\Theta$ is open, we can find a neighbourhood $U_\theta$ of $\theta$ on which $f(\vartheta,\omega) \leq f(\theta,\omega) + 1$ for all $\vartheta \in U_\theta$.
Take $h_\theta(\omega) := |f(\theta,\omega)| + 1$, recalling that the absolute value of a measurable function is measurable and sums of measurable functions are measurable.
Then $\int_\Omega h_\theta(\omega) \mathrm{d}\mu(\omega) = \int_\Omega |f(\theta,\omega)| \mathrm{d}\mu(\omega) + 1 < \infty$ and $|f(\theta,\omega)| \leq h_\theta(\omega)$, as required.
\end{proof}

\begin{lemma}[Differentiate under the integral] \label{lem: interchange}
Let $\Omega$ be a measurable space equipped with a measure $\mu$, let $\Theta \subseteq \mathbb{R}^d$ be an open set and let $f : \Theta \times \Omega \rightarrow \mathbb{R}$ be a Carath\'{e}odory function.
Assume further that $f$ is locally uniformly integrably bounded and that, for each $i$ and each $\omega$, the function $\theta \mapsto \partial_{\theta_i} f(\theta,\omega)$ is locally uniformly integrably bounded.
Then the function $g : \Theta \rightarrow \mathbb{R}$ defined by
$$
g(\theta) := \int_\Omega f(\theta,\omega) \mathrm{d}\mu(\omega)
$$
is continuously differentiable and
$$
\nabla_\theta g(\theta) = \int_\Omega \nabla_\theta f(\theta,\omega) \mathrm{d}\mu(\omega).
$$
\end{lemma}
\begin{proof}
This standard result can be found, for example, in \citet[][Theorem 24.5, p.193]{aliprantis1998principles}, \citet[][Theorem 16.8, pp.181–182]{billingsley1979probability}, with the statement here based on the account in \cite{border2016differentiating}.
\end{proof}

The proof of \Cref{prop: gradients} can now be presented:

\begin{proof}[Proof of \Cref{prop: gradients}]
Using \eqref{eq:KSD compute} and the reparametrisation trick \citep{glynn1986stochastic,ecuyer1995,kingma2013}:
\begin{align}
    \nabla_\theta \left[ \mathcal{D}_{\textsc{S}}(P,T_\#^\theta Q)^2 \right] &= \nabla_\theta \E_{Y,Y' \sim T_\#^\theta Q} \left[ u_p(Y,Y') \right] \nonumber \\
    &= \nabla_\theta \E_{X,X' \sim Q} \left[ u_p(T^\theta(X),T^\theta(X')) \right]  \label{eq: waiting to change}
\end{align}
From \Cref{lem: interchange}, the preconditions of \Cref{prop: gradients} justify the interchange of the derivative and the expectation in \eqref{eq: waiting to change}.
Indeed, in the setting of \Cref{lem: interchange} we identify $\Omega = \mathcal{X} \times \mathcal{X}$, $\mu = Q \times Q$ and $f(\theta,\omega) = u_p(T^\theta(x),T^\theta(x'))$ where $\omega = (x,x')$.
That $f$ is a Carath\'{e}odory function follows from (A1) and (A4) of \Cref{prop: gradients}, where we note that (A4) implies $\theta \mapsto \nabla_\theta u_p(T^\theta(x),T^\theta(x'))$ is continuous.
That $f$ is locally uniformly integrably bounded follows from assumptions (A2) and (A4) together with \Cref{lem: LUIB}.
Similarly, that $\partial_{\theta_i} f$ is locally uniformly integrably bounded follows from assumptions (A3) and (A4) together with \Cref{lem: LUIB}.
Thus the preconditions of \Cref{lem: interchange} hold.

Interchanging the derivative with the expectation gives that
\begin{align*}
    \nabla_\theta \left[ \mathcal{D}_{\textsc{S}}(P,T_\#^\theta Q)^2 \right]  &= \E_{X,X' \sim Q} \left[ \nabla_\theta u_p(T^\theta(X),T^\theta(X')) \right] \\
    &= \E\left[ \frac{1}{n(n-1)} \sum_{i \neq j}^n \nabla_\theta u_p(T^\theta(x_i),T^\theta(x_j)) \right] ,
\end{align*}
as claimed.
\end{proof}

Note that we presented stronger conditions in \Cref{prop: gradients} than are required.
This was to control the length of the main text, but it is immediately clear from the proof of \Cref{prop: gradients} that these conditions can be weakened to those that are required for \Cref{lem: interchange} to hold.

\subsection{Proof of \Cref{prop: transport map properties}} \label{app: T prop}

First we present an existence result in \Cref{prop:transport map existence}, before considering regularity of the associated transport map.
Recall that in this paper all measurable spaces $\mathcal{X}$ and $\mathcal{Y}$ are equipped with their respective Borel $\sigma$-algebras $\Sigma_{\mathcal{X}}$ and $\Sigma_{\mathcal{Y}}$.
A separable complete metric space equipped with its Borel $\sigma$-algebra is called a \emph{standard Borel space}.
A measure $Q \in \mathcal{P}(\mathcal{X})$ is said to be \emph{continuous} if $Q(\{x\}) = 0$ for all $x \in \mathcal{X}$.
A map $f : \mathcal{X} \rightarrow \mathcal{Y}$ is called a \emph{Borel isomorphism} if $f$ is a bijection and both $f$ and $f^{-1}$ are Borel measurable.
A fundamental result that we will exploit is known as the \emph{isomorphism theorem for measures}:

\begin{theorem}[Isomorphism Theorem] \label{thm: iso}
Let $\mathcal{X}$ be a standard Borel space and $Q\in\mathcal{P}(\mathcal{X})$ be continuous. 
Then there is a Borel isomorphism $f:\mathcal{X}\rightarrow[0,1]$ with $f_\#Q = m |_{[0,1]}$, where $m |_{[0,1]}$ is the Lebesgue measure restricted to $[0,1]$.
\end{theorem}
\begin{proof}
This result can be found as Theorem 17.41 in \cite{kechrisclassicaldescriptivesettheory}.
\end{proof}

\begin{proposition} \label{prop:transport map existence}
Suppose that $\mathcal{X}$ and $\mathcal{Y}$ are separable complete metric spaces and suppose that $Q\in\mathcal{P}(\mathcal{X})$ and $P\in\mathcal{P}(\mathcal{Y})$ are such that $Q(\{x\}) = 0$ and $P(\{y\}) = 0$ for all $x\in\mathcal{X}$ and $y\in\mathcal{Y}$. Then there exists a measurable function $T : \mathcal{X} \rightarrow \mathcal{Y}$ such that $T_\# Q = P$.
\end{proposition}
\begin{proof}
Our assumptions imply that $\mathcal{X}$, $\mathcal{Y}$ are standard Borel spaces and $Q \in \mathcal{P}(\mathcal{X})$, $P \in \mathcal{P}(\mathcal{Y})$ are continuous.
Thus from \Cref{thm: iso}, there exists a Borel isomorphism $f:\mathcal{X}\rightarrow [0,1]$ such that $f_\# Q = m|_{[0,1]}$ and a Borel isomorphism $g:\mathcal{Y}\rightarrow [0,1]$ such that $g_\# P = m|_{[0,1]}$. 
Then $T := g^{-1} \circ f : \mathcal{X} \rightarrow \mathcal{Y}$ is measurable and satisfies $T_\# Q = P$, as required.
\end{proof}

Now we can present the proof of \Cref{prop: transport map properties}.

\begin{proof}[Proof of \Cref{prop: transport map properties}]

\Cref{assumption reference} ensures that $\mathcal{X}$ is a separable complete metric space and $Q(\{x\}) = 0$ for all $x \in \mathcal{X}$.
\Cref{assumption target} restricts attention to $\mathcal{Y} = \mathbb{R}^d$, meaning that $\mathcal{Y}$ is a separable complete metric space, and requires $P$ to admit a density on $\mathcal{Y}$, meaning that $P(\{y\}) = 0$ for all $y \in \mathcal{Y}$.
Thus the existence of a transport map $T$ from $Q$ to $P$ is guaranteed by \Cref{prop:transport map existence}. 

It remains to show that, for \emph{any} such transport map, $T \in \prod_{i=1}^d L^2(Q)$.
To this end, we have that
\begin{align*}
    \|T\|_{\prod_{i=1}^d L^2(Q)}^2 = \sum_{i=1}^d \|T_i\|_{L^2(Q)}^2 = \sum_{i=1}^d \int_{\mathcal{X}} T_i(x)^2 \,\mathrm{d}Q(x)
    &= \int_{\mathcal{X}} \|T(x)\|^2 \,\mathrm{d}Q(x) \\
    &\stackrel{(*)}{=} \int_{\mathbb{R}^d} \|x\|^2 \,\mathrm{d}(T_\#Q)(x) 
    = \int_{\mathbb{R}^d} \|x\|^2 \,\mathrm{d}P(x) < \infty,
\end{align*}
where a change of variables was used at $(*)$ and the final inequality follows from the assumption that $P\in\mathcal{P}_2(\mathbb{R}^d)$ in \Cref{assumption target}. 
\end{proof}

\subsection{Proof of \Cref{theorem: convergence theorem}} \label{app: thm pf}

Recall that for $P,P' \in \mathcal{P}_1(\mathbb{R}^d)$ the (first) \emph{Wasserstein  distance} is defined as \citep[][Remark 6.5]{villani_2009}
\begin{equation}
    W_1(P,P') := \sup_{f \in \mathcal{F}}\left|\E_{Y\sim P}[f(Y)] - \E_{Y\sim P'}[f(Y)]\right|, \label{eq: W1}
\end{equation}
where $\mathcal{F} := \left\{f:\mathbb{R}^d\rightarrow\mathbb{R}\,|\, \| f \|_L\leq 1\right\}$ and $\|f\|_L := \sup_{x\neq y} \frac{|f(x)-f(y)|}{\|x-y\|}$ is the Lipschitz seminorm on $\mathbb{R}^d$.
Our proof of \Cref{theorem: convergence theorem} is based on the following result that relates convergence in $W_1$ to convergence in KSD:

\begin{proposition}[Wasserstein Controls KSD] \label{prop: gorham wasserstein controls KSD}
Let $k : \mathbb{R}^d \times \mathbb{R}^d \rightarrow \mathbb{R}$ be symmetric positive definite with $(x,y) \mapsto k(x,y)$, $(x,y) \mapsto \partial_{x_i} \partial_{y_i} k(x,y)$ and $(x,y) \mapsto \partial_{x_i} \partial_{x_j} \partial_{y_i} \partial_{y_j} k(x,y)$ continuous and bounded for all $i,j \in \{1,\dots,d\}$.
Let $P \in \mathcal{P}(\mathbb{R}^d)$ admit a density function $p$ such that $\nabla\log p$ is Lipschitz with $\E_{X\sim P}[\|\nabla\log p(X)\|_2^2] < \infty$. 
Let $\mathcal{D}_{\textsc{S}}$ denote the KSD based on $P$ and $k$, as defined in \eqref{eq: KSD final}.
Then a sequence $(Q_n)_{n\in\mathbb{N}} \subset \mathcal{P}(\mathbb{R}^d)$ satisfies $\mathcal{D}_{\textsc{S}}(P,Q_n) \rightarrow 0$ whenever $W_1(P,Q_n)\rightarrow 0$.
\end{proposition}
\begin{proof}
This result is Proposition 9 of \cite{gorham2017measuring}.
\end{proof}

Recall that in this paper $\mathcal{X}$ is always assumed to be a Borel space.
The following result is also required:

\begin{lemma}[$L^2$ Controls Wasserstein] \label{lem: l2 was}
Let $Q \in \mathcal{P}(\mathcal{X})$ and let $S,T \in \prod_{i=1}^d L^2(Q)$.
Then we have the bound $W_1(S_\# Q, T_\# Q) \leq \|S - T\|_{\prod_{i=1}^d L^2(Q)}$.
\end{lemma}
\begin{proof}
From the definition of the (first) Wasserstein distance, we have
\begin{align*}
    W_1(S_\# Q, T_\# Q) &= \sup_{\|f\|_L \leq 1}\left|\int_{\mathcal{X}} f(x)\,\mathrm{d}S_\#Q(x) -\int_{\mathcal{X}} f(x)\,\mathrm{d}T_\#Q(x)   \right|  \\ 
    &= \sup_{\|f\|_L \leq 1}\left|\int_{\mathcal{X}} f(S(x)) - f(T(x))\,\mathrm{d}Q(x)   \right| .
\end{align*}
If $\|f\|_L \leq 1$ then $|f(a) - f(b)|\leq \|a-b\|$ for all $a,b \in \mathbb{R}^d$, and so
\begin{align*}
     W_1(S_\# Q, T_\# Q) & \leq \int_{\mathcal{X}} \|S(x) - T(x)\|\,\mathrm{d}Q(x) \\
     & \leq \left( \int_{\mathcal{X}} \|S(x)-T(x)\|^2 \; \mathrm{d}Q(x) \right)^{1/2}
     = \|S - T\|_{\prod_{i=1}^d L^2(Q)}
\end{align*}
where the second inequality is Jensen's inequality.
\end{proof}

Our final ingredient is a basic result on the inverse multi-quadric kernel:

\begin{lemma}[Derivatives of the Inverse Multi-quadric Kernel] \label{lem: IMQ bd}
The inverse multi-quadric kernel in \eqref{eq: IMQ}, $k(x,y) = (c^2 + \|x-y\|^2)^\beta$, with $c > 0$ and $\beta \in (-1,0)$, satisfies
\begin{equation*}
    \sup_{x,y \in \mathbb{R}^d} \left| \partial_{x_1}^{\alpha_1} \dots \partial_{x_d}^{\alpha_d} \partial_{y_1}^{\alpha_1} \dots \partial_{y_d}^{\alpha_d} \; k(x,y) \right| < \infty
\end{equation*}
for all multi-indices $\alpha = (\alpha_1,\dots,\alpha_d) \in \mathbb{N}_0^d$.
\end{lemma}
\begin{proof}
For $\alpha \in \mathbb{N}_0^d$ let $|\alpha| := \alpha_1 + \dots + \alpha_d$.
Recall that a polynomial $\prod_{|\alpha| \leq s} c_\alpha z_1^{\alpha_1} \dots z_d^{\alpha_d}$ is said to have \textit{maximal degree} $s$, where $s = |\alpha|$ is the largest integer for which $c_\alpha \neq 0$ for some $\alpha \in \mathbb{N}_0^d$.
Let 
$$
\mathcal{F} := \left\{ (x,y) \mapsto k(x,y) \frac{r_m(x-y)}{ (c^2 + \|x-y\|^2)^m } : r_m \text{ is a polynomial of maximal degree } 2m, \; m \in \mathbb{N}_0 \right\} .
$$
Then $k(x,y) \in \mathcal{F}$ and $\mathcal{F}$ is closed under the action of each of the differential operators $\partial_{x_i}\partial_{y_i}$, $i \in \{1,\dots,d\}$.
Indeed, we have from the product rule that
\begin{align*}
\textstyle \partial_{x_i} \left[ k(x,y) \frac{r_m(x-y)}{ (c^2 + \|x-y\|^2)^m } \right] & = \textstyle \frac{2 \beta (x_i - y_i) k(x,y) r_m(x-y)}{(c^2 + \|x-y\|^2)^{m+1}} + \frac{ k(x,y) \partial_{x_i} r_m(x-y)}{ (c^2 + \|x-y\|^2)^m } - \frac{2 m (x_i - y_i) k(x,y) r_m(x-y)}{ (c^2 + \|x-y\|^2)^{m+1} } 
\end{align*}
and
\begin{align}
\textstyle \partial_{x_i} \partial_{y_i} \left[ k(x,y) \frac{r_m(x-y)}{ (c^2 + \|x-y\|^2)^m } \right] & = \textstyle \left[ - \frac{2 \beta k(x,y) r_m(x-y)}{(c^2 + \|x-y\|^2)^{m+1}} - \frac{4 \beta^2 (x_i - y_i)^2 k(x,y) r_m(x-y)}{(c^2 + \|x-y\|^2)^{m+2}} \right. \nonumber \\
& \textstyle \qquad \left. + \frac{2 \beta (x_i - y_i) k(x,y) \partial_{y_i} r_m(x-y)}{(c^2 + \|x-y\|^2)^{m+1}} + \frac{4 (m+1) \beta (x_i - y_i)^2 k(x,y) r_m(x-y)}{(c^2 + \|x-y\|^2)^{m+2}} \right] \nonumber \\
& \textstyle + \left[ - \frac{ 2 \beta (x_i - y_i) k(x,y) \partial_{x_i} r_m(x-y)}{ (c^2 + \|x-y\|^2)^{m+1} } + \frac{ k(x,y) \partial_{x_i} \partial_{y_i} r_m(x-y)}{ (c^2 + \|x-y\|^2)^m } \right. \nonumber  \\
& \textstyle \qquad \left. + \frac{ 2 m (x_i - y_i) k(x,y) \partial_{x_i} r_m(x-y)}{ (c^2 + \|x-y\|^2)^{m+1} } \right] \nonumber  \\
& \textstyle + \left[ \frac{2 m k(x,y) r_m(x-y)}{ (c^2 + \|x-y\|^2)^{m+1} } + \frac{4 \beta m (x_i - y_i)^2 k(x,y) r_m(x-y)}{ (c^2 + \|x-y\|^2)^{m+2} } \right. \nonumber  \\
& \textstyle \qquad \left. - \frac{2 m (x_i - y_i) k(x,y) \partial_{y_i} r_m(x-y)}{ (c^2 + \|x-y\|^2)^{m+1} } + \frac{4 m (m+1) (x_i - y_i)^2 k(x,y) r_m(x-y)}{ (c^2 + \|x-y\|^2)^{m+2} } \right] \nonumber  \\
& = \textstyle k(x,y) \frac{r_{m+2}(x-y)}{(c^2 + \|x-y\|^2)^{m+2}} \label{eq: big expansion}
\end{align}
where $r_{m+2}(x-y)$ has been implicitly defined.
Since $\partial_{x_i} (x_i - y_i)^s = s(x_i - y_i)^{s-1}$, it follows that the terms $\partial_{x_i} r_m(x-y)$, $\partial_{y_i} r_m(x-y)$ and $\partial_{x_i} \partial_{y_i} r_m(x-y)$ appearing in \eqref{eq: big expansion} are polynomials in $x-y$ of maximal degree $2m$.
Thus, from \eqref{eq: big expansion}, $r_{m+2}(x-y)$ is a polynomial of maximal degree $2(m+2)$, showing that the set $\mathcal{F}$ is closed under the action of $\partial_{x_i} \partial_{y_i}$.

Since the differential operator $\partial_{x_1}^{\alpha_1} \dots \partial_{x_d}^{\alpha_d}$ is obtained by repeated application of operators of the form $\partial_{x_i} \partial_{y_i}$, and since it is clear that all elements of $\mathcal{F}$ are bounded on $\mathbb{R}^d \times \mathbb{R}^d$, the claim is established.
\end{proof}

Now we can prove \Cref{theorem: convergence theorem}:

\begin{proof}[Proof of \Cref{theorem: convergence theorem}]

First note that our preconditions are a superset of those required for \Cref{theorem:gorham inverse multi quadric}.
Thus the conclusion of \Cref{theorem:gorham inverse multi quadric} holds; namely, if $\mathcal{D}_{\text{S}}(P,(T_n)_\#Q)\rightarrow 0$ then $(T_n)_\#Q\Rightarrow P$. 
From \eqref{eq: main result} we have that $\mathcal{D}_{\text{S}}(P,(T_n)_\#Q)$ and $\inf_{T\in\mathcal{T}_n}\mathcal{D}_{\text{S}}(P,T_\#Q)$ agree in the $n \rightarrow \infty$ limit.
Thus it is sufficient to show that $\inf_{T\in\mathcal{T}_n}\mathcal{D}_{\text{S}}(P,T_\#Q) \rightarrow 0$ in the $n \rightarrow \infty$ limit.

Second, note that our preconditions are also a superset of those required for \Cref{prop: gorham wasserstein controls KSD}. 
Indeed, from \Cref{lem: IMQ bd} the inverse multi-quadric kernel in \eqref{eq: IMQ} is infinitely differentiable with derivatives of all orders bounded.
Thus it is sufficient to show that $\inf_{T\in\mathcal{T}_n}W_1(P,T_\#Q) \rightarrow 0$ in the $n \rightarrow \infty$ limit.



From \Cref{assumption reference,assumption target} and \Cref{prop: transport map properties}, there exists a map $T \in \prod_{i=1}^d L^2(Q)$ with $T_\# Q = P$.
From \Cref{assum: G contains T}, there is a set $\textgoth{T}$ such that $T \in \textgoth{T} \subseteq \prod_{i=1}^d L^2(Q)$ and the set $\mathcal{T}_\infty$ is dense in $\textgoth{T}$.
Thus there exists a sequence $(S_n)_{n \in \mathbb{N}} \subset \mathcal{T}_\infty$ with $S_n \rightarrow T$ in $\prod_{i=1}^d L^2(Q)$.

For each $n \in \mathbb{N}$, let $m_n \in \mathbb{N}$ denote the smallest integer $m$ for which $S_n \in \mathcal{T}_m$, which is well-defined since $S_n \in \mathcal{T}_\infty = \cup_{i \in \mathbb{N}} \mathcal{T}_i$.
Let $M_n := \max\{m_1,\dots,m_n, n\}$, so that $M_n$ is a non-decreasing sequence with $M_n \rightarrow \infty$ in the $n \rightarrow \infty$ limit. 
Note that, since $\mathcal{T}_i \subseteq \mathcal{T}_j$ for all $i \leq j$, we have $S_n \in \mathcal{T}_{M_n}$.

Thus from \Cref{lem: l2 was} we conclude that
\begin{align}
    0 \leq \inf_{T \in \mathcal{T}_{M_n}} W_1(P,T_\#Q) \leq W_1(P,(S_n)_\#Q) \leq \|T - S_n\|_{\prod_{i=1}^d L^2(Q)} \rightarrow 0 \label{eq: converge subseq}
\end{align}
in the $n \rightarrow \infty$ limit.
Again, since $\mathcal{T}_i \subseteq \mathcal{T}_j$ for $i \leq j$, the sequence $n \mapsto \inf_{T \in \mathcal{T}_n} W_1(P,T_\#Q)$ is non-increasing and, from \eqref{eq: converge subseq}, it has a subsequence that converges to 0. 
It follows that $\lim_{n \rightarrow \infty} \inf_{T \in \mathcal{T}_n} W_1(P,T_\#Q) = 0$, as required.




\end{proof}

\subsection{Proof of \Cref{prop: ReLU DNN}} \label{app: DNN}

Recall the \emph{rectified linear unit} activation function $\sigma(x) = \max(0,x)$, which we consider to be applied componentwise when $x \in \mathbb{R}^d$.

\begin{definition}[Deep ReLU Neural Network] \label{def: DNN}
A \emph{deep ReLU neural network} with $l$ \emph{hidden layers} from $\mathbb{R}^p$ to $\mathbb{R}^d$ is a function $f : \mathbb{R}^p \rightarrow \mathbb{R}^d$ of the form
$$
f = F_{l+1} \circ \sigma \circ F_l \circ \dots \circ F_2 \circ \sigma \circ F_1
$$
where $F_i : \mathbb{R}^{w_{i-1}} \rightarrow \mathbb{R}^{w_i}$, $i = 1,\dots,l$, is an affine transformation, $F_{l+1} : \mathbb{R}^{w_l} \rightarrow \mathbb{R}^{w_{l+1}}$ is a linear transformation, $w_0 = p$ is the \emph{input dimension}, $w_{l+1} = d$ is the \emph{output dimension}, and $w_i \in \mathbb{N}$, $i = 1,\dots,l$, is the \emph{width} of the $i$th \emph{hidden layer}.
The set of all deep ReLU neural networks with $l$ \emph{hidden layers} from $\mathbb{R}^p$ to $\mathbb{R}^d$ with maximum width $\max\{w_1,\dots,w_l\} \leq n$ is denoted $\mathcal{R}_{l,n}(\mathbb{R}^p\rightarrow\mathbb{R}^d)$ and we let $\mathcal{R}_{l,\infty}(\mathbb{R}^p\rightarrow\mathbb{R}^d) := \lim_{n \rightarrow \infty} \mathcal{R}_{l,n}(\mathbb{R}^p\rightarrow\mathbb{R}^d)$.
\end{definition}

The following, essentially trivial observation will be useful:
\begin{proposition} \label{prop: wide NN}
$\prod_{i=1}^d\mathcal{R}_{l,\infty}(\mathbb{R}^p\rightarrow\mathbb{R}) \subset \mathcal{R}_{l,\infty}(\mathbb{R}^p\rightarrow\mathbb{R}^d)$. 
\end{proposition}
\begin{proof}
For fixed $n \in \mathbb{N}$, there is a canonical injection from $\prod_{i=1}^d\mathcal{R}_{l,n}(\mathbb{R}^p\rightarrow\mathbb{R})$ into $\mathcal{R}_{l,nd}(\mathbb{R}^p\rightarrow\mathbb{R}^d)$ that concatenates the $d$ neural networks width-wise, to form a single neural network with width $nd$. 
Since every element of $\prod_{i=1}^d\mathcal{R}_{l,\infty}(\mathbb{R}^p\rightarrow\mathbb{R})$ belongs to $\prod_{i=1}^d\mathcal{R}_{l,n}(\mathbb{R}^p\rightarrow\mathbb{R})$ for a sufficiently large $n \in \mathbb{N}$, the claim is established.
\end{proof}

Here we introduce the shorthand $L^2(\mathbb{R}^p)$ for $L^2(\lambda_{\mathbb{R}^p})$ where $\lambda_{\mathbb{R}^p}$ is the Lebesgue measure on $\mathbb{R}^p$.
The following result on the approximation properties of deep ReLU neural networks, which derives from the fact that the set of continuous piecewise linear functions $f : \mathbb{R} \rightarrow \mathbb{R}$ is dense in $L^2(\mathbb{R})$, will be required.

\begin{proposition} \label{prop: dense NN}
For every function $f \in \prod_{i=1}^d L^2(\mathbb{R}^p)$ and every $\epsilon > 0$, there exists a function $g \in  \mathcal{R}_{l, \infty}(\mathbb{R}^p\rightarrow\mathbb{R}^d)$ such that $\|f - g\|_{\prod_{i=1}^d L^2(\mathbb{R}^p)} < \epsilon$, where $l = \lceil \log_2(p+1) \rceil$.
\end{proposition}
\begin{proof}
For $d=1$, this result is a special case of Theorem 2.3 in \cite{arora2016understanding}, which derives from the fact that continuous piecewise linear functions are dense in $L^2(\mathbb{R})$.

For general $d$, we observe that for each component $f_i \in L^2(\mathbb{R}^p)$ we can find a function $g_i \in \mathcal{R}_{l,\infty}(\mathbb{R}^p\rightarrow\mathbb{R})$ with $\|f_i - g_i\|_{L^2(\mathbb{R}^p)} < \epsilon / \sqrt{d}$.
Then, letting $g = (g_1,\dots,g_d) : \mathbb{R}^p \rightarrow \mathbb{R}^d$, we have that $g \in \prod_{i=1}^d \mathcal{R}_{l,\infty}(\mathbb{R}^p\rightarrow\mathbb{R})$ and
$$
\|f - g\|_{\prod_{i=1}^d L^2(\mathbb{R}^p)} = \sqrt{ \sum_{i=1}^d \|f_i - g_i\|_{L^2(\mathbb{R}^p)}^2 } < \sqrt{ \sum_{i=1}^d \frac{\epsilon^2}{d} } = \epsilon . 
$$
Finally, we note from \Cref{prop: wide NN} that $g \in \prod_{i=1}^d\mathcal{R}_{l,\infty}(\mathbb{R}^p\rightarrow\mathbb{R}) \subset \mathcal{R}_{l,\infty}(\mathbb{R}^p\rightarrow\mathbb{R}^d)$. 
\end{proof}

Now we present the proof of \Cref{prop: ReLU DNN}:

\begin{proof}[Proof of \Cref{prop: ReLU DNN}]

From \Cref{assumption reference,assumption target} and \Cref{prop: transport map properties} there exists $T \in \prod_{i=1}^d L^2(Q)$ such that $T_\# Q = P$.
Let $\textgoth{T} = \prod_{i=1}^d L^2(Q)$, so that $T \in \textgoth{T}$ is satisfied.

From the statement of \Cref{prop: ReLU DNN} we have $\mathcal{T}_n := \mathcal{R}_{l,n}(\mathbb{R}^p\rightarrow \mathbb{R}^d)$ with $l = \lceil \log_2(p+1) \rceil$.
From \Cref{def: DNN} it is therefore clear that $\mathcal{T}_n \subseteq \mathcal{T}_m$ whenever $n \leq m$ and that $\mathcal{T}_\infty = \mathcal{R}_{l,\infty}(\mathbb{R}^p\rightarrow \mathbb{R}^d)$. 

Thus all parts of \Cref{assum: G contains T} have been verified except the part that requires $\mathcal{T}_\infty$ to be dense in $\textgoth{T}$; i.e. that the set $\mathcal{R}_{l,\infty}(\mathbb{R}^p\rightarrow\mathbb{R}^d)$ is dense in the Hilbert space $\prod_{i=1}^d L^2(Q)$.
To establish this last part, we will make use of \Cref{prop: dense NN}:

Let $T \in \prod_{i=1}^d L^2(Q)$ and $\epsilon > 0$.
From the definition of $L^2(Q)$, there exists $c \geq 0$ such that, for each of the coordinates $i \in \{1,\dots,d\}$,
\begin{equation*}
    \int_{\mathbb{R}^p\setminus[-c,c]^p} T_i(x)^2\,\mathrm{d}Q(x) < \frac{\epsilon}{4d}.
\end{equation*}
Let 
\begin{equation*}
    f_i(x) := \begin{cases}
    T_i(x), & x\in[-c,c]^p, \\
    0, & x\in \mathbb{R}^p\setminus[-c,c]^p .
    \end{cases}
\end{equation*}
Our assumption that $Q$ admits a positive and continuous density $q$ on $\mathbb{R}^p$ ensures that $f\in \prod_{i=1}^d L^2(\mathbb{R}^p)$, since
\begin{align}
    \|f\|_{\prod_{i=1}^d L^2(\mathbb{R}^p)}^2 
    = \sum_{i=1}^d \int_{[-c,c]^p} T_i(x)^2 \mathrm{d}x 
    = \sum_{i=1}^d \int_{[-c,c]^p} \frac{T_i(x)^2}{q(x)} \mathrm{d}Q(x) & \leq \left[ \sup_{x \in [-c,c]^d} \frac{1}{q(x)} \right] \sum_{i=1}^d \int_{[-c,c]^p} T_i(x)^2 \mathrm{d}Q(x) \nonumber \\
    & \leq \left[ \sup_{x \in [-c,c]^p} \frac{1}{q(x)} \right] \sum_{i=1}^d \|T_i\|_{L^2(Q)}^2 \nonumber \\
    & = \left[ \sup_{x \in [-c,c]^p} \frac{1}{q(x)} \right]  \|T\|_{\prod_{i=1}^d L^2(Q)}^2, \label{eq: L2 bound}
\end{align}
where the supremum in \eqref{eq: L2 bound} is finite, since $q^{-1}$ is well-defined and continuous on the compact set $[-c,c]^p$.
Let also $q_{\max} := \sup_{x \in \mathbb{R}^p} q(x)$, which is well-defined since we assumed $q$ to be continuous and bounded on $\mathbb{R}^p$.
Then, since $f\in \prod_{i=1}^d L^2(\mathbb{R}^p)$, we may evoke \Cref{prop: dense NN} to find a function $g \in \mathcal{R}_{l,\infty}(\mathbb{R}^p \rightarrow \mathbb{R}^d)$ such that $\|f-g\|_{\prod_{i=1}^d L^2(\mathbb{R}^p)}^2 < \epsilon / (4 q_{\max})$.
It remains to check that $g$ approximates $T$ in $\prod_{i=1}^d L^2(Q)$.
To this end, we can use the triangle inequality in $\prod_{i=1}^d L^2(Q)$ and the fact that $(a+b)^2 \leq 2(a^2 + b^2)$ to see that
\begin{align*}
    \|T - g\|_{\prod_{i=1}^d L^2(Q)}^2 & \leq 2 \|T - f\|_{\prod_{i=1}^d L^2(Q)}^2 + 2\|f - g\|_{\prod_{i=1}^d L^2(Q)}^2 \\
    & = 2 \sum_{i=1}^d \int_{\mathbb{R}^p \setminus [-c,c]^p} T_i(x)^2 \mathrm{d}Q(x) + 2 \sum_{i=1}^d \int (f_i(x) - g_i(x))^2 q(x) \mathrm{d}x  \\
    & \leq 2 \sum_{i=1}^d \int_{\mathbb{R}^p \setminus [-c,c]^p} T_i(x)^2 \mathrm{d}Q(x) + 2 q_{\max} \sum_{i=1}^d \int (f_i(x) - g_i(x))^2 \mathrm{d}x < \frac{\epsilon}{2} + \frac{\epsilon}{2} = \epsilon .
\end{align*}
Since $\epsilon > 0$ was arbitrary, this argument shows that the set $\mathcal{R}_{l,\infty}(\mathbb{R}^p \rightarrow \mathbb{R}^d)$ is dense in $\prod_{i=1}^d L^2(Q)$, as required.
\end{proof}

\section{Computational Details} \label{sec: computational details}

This section provides full details for the experiments presented in \Cref{sec:experiments}.

\subsection{Performance Metric} \label{subsec: performance metrics}

To estimate the Wasserstein-1 distance between the target distribution and approximations we computed the \emph{earth mover distance} (EMD) between two uniformly weighted empirical measures, each formed from $10^4$ samples from their respective distributions and with the Euclidean distance between the samples used to construct the cost matrix. 
The EMD was computed using an implementation in the Python Optimal Transport (\texttt{POT}) package \citep{flamary2017pot}. 
For the target distribution, independent samples were used in the synthetic test bed in \Cref{subsec: toy examples} and thinned samples from a long HMC chain were used for the real examples in \Cref{subsec: biochemical oxygen,subsec: lotka-volterra}. 
For the approximate distribution, independent samples from $T^\theta_\#Q$ were used.

\subsection{Details of the Synthetic Test Bed} \label{subsec: details of toy models}

To assess the proposed methods, we considered the following bivariate densities
\begin{align*}
    p_1(x,y) &:= \mathcal{N}(x; 0, \eta_1^2)\mathcal{N}(y; \sin(bx), \eta_2^2), \\
    p_2(x,y) &:= \mathcal{N}(x; 0, \sigma_1^2)\mathcal{N}(y; ax^2, \sigma_2^2) , \\
    p_3(x,y) &:= \textstyle \frac{1}{n}\sum_{i=1}^n \mathcal{N}(x,y; \mu_i, \sigma^2I_2) ,
\end{align*}
where $\mathcal{N}(x;\mu,\sigma^2)$ is the univariate Gaussian density with mean $\mu$ and variance $\sigma^2$, and $\mathcal{N}(x,y;\mu,K)$ is the bivariate Gaussian density with mean vector $\mu$ and covariance matrix $K$. The parameter choices for the sinusoidal experiment $p_1$ were $\eta_1^2 = 1.3^2, \eta_2^2 = 0.001^2$ and $b=1.2$. The parameter choices for the banana experiment $p_2$ were $\sigma_1^2 = 1, \sigma_2^2 = 0.1^2$ and $a = 0.5$. The parameter choices for the multi-modal experiment $p_3$ were $n = 4, \mu_1 = (1,1),\mu_2 = (1,-1),\mu_3 = (-1,-1), \mu_4 = (-1,1)$ and $\sigma^2 = 0.2^2$. The target densities can be seen in \Cref{fig: synthetic test bed densities}.

\begin{figure*}[h!] 
    \centering
    \includegraphics[width = \textwidth]{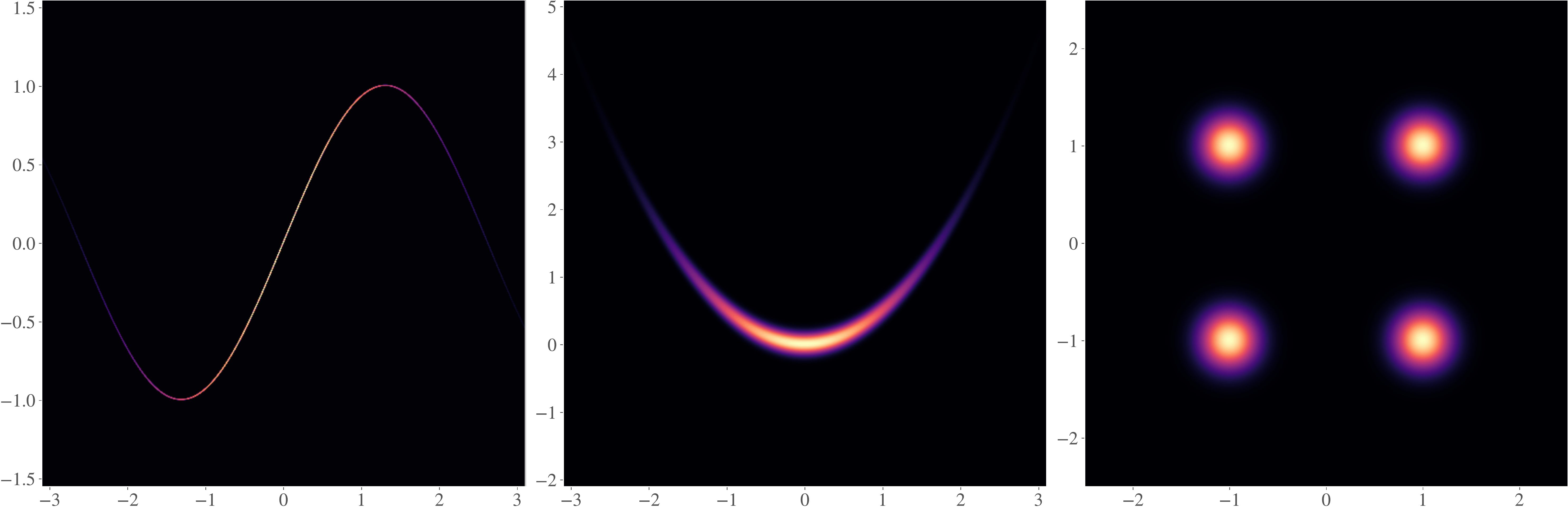}
    \caption{Contour plots of the three synthetic test densities $p_1$, $p_2$ and $p_3$ from left to right used in \Cref{subsec: toy examples}.}
    \label{fig: synthetic test bed densities}
\end{figure*}

Computational costs for each method in terms of number of target evaluations and CPU wall-clock time against performance are shown, respectively, in \Cref{fig: wass vs evaluations} and \Cref{fig: wass vs time}. 
From \Cref{fig: wass vs evaluations}, there is no clear sense in which KSD or KLD out-performs the other across the different synthetic tests; this is in line with the conclusion of \Cref{subsec: toy examples}.
For CPU wall-clock time in \Cref{fig: wass vs time}, KSD-based measure transport is approximately three to five times slower than its KLD counterpart. However, note that our implementation of KSD is not production code and further performance gains can certainly be achieved. 

\begin{figure}[h!] 
   \centering
   \includegraphics[width=\textwidth]{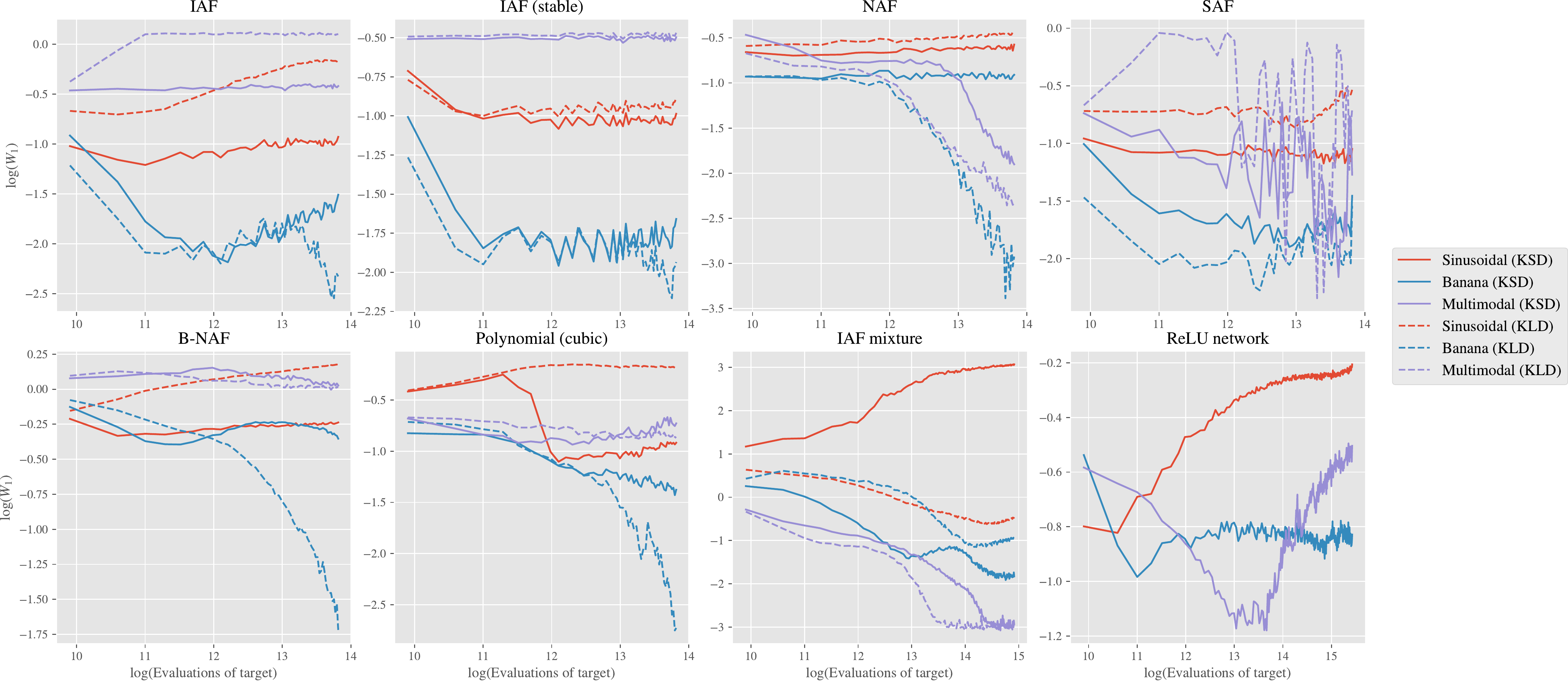}

\caption{
    Wasserstein-1 metric, $W_1$, as a function of the total number of evaluations of either $p$ or its gradient, for each synthetic test experiment.
}
\label{fig: wass vs evaluations}
\end{figure}

\begin{figure}[h!] 
   \centering
   \includegraphics[width=\textwidth]{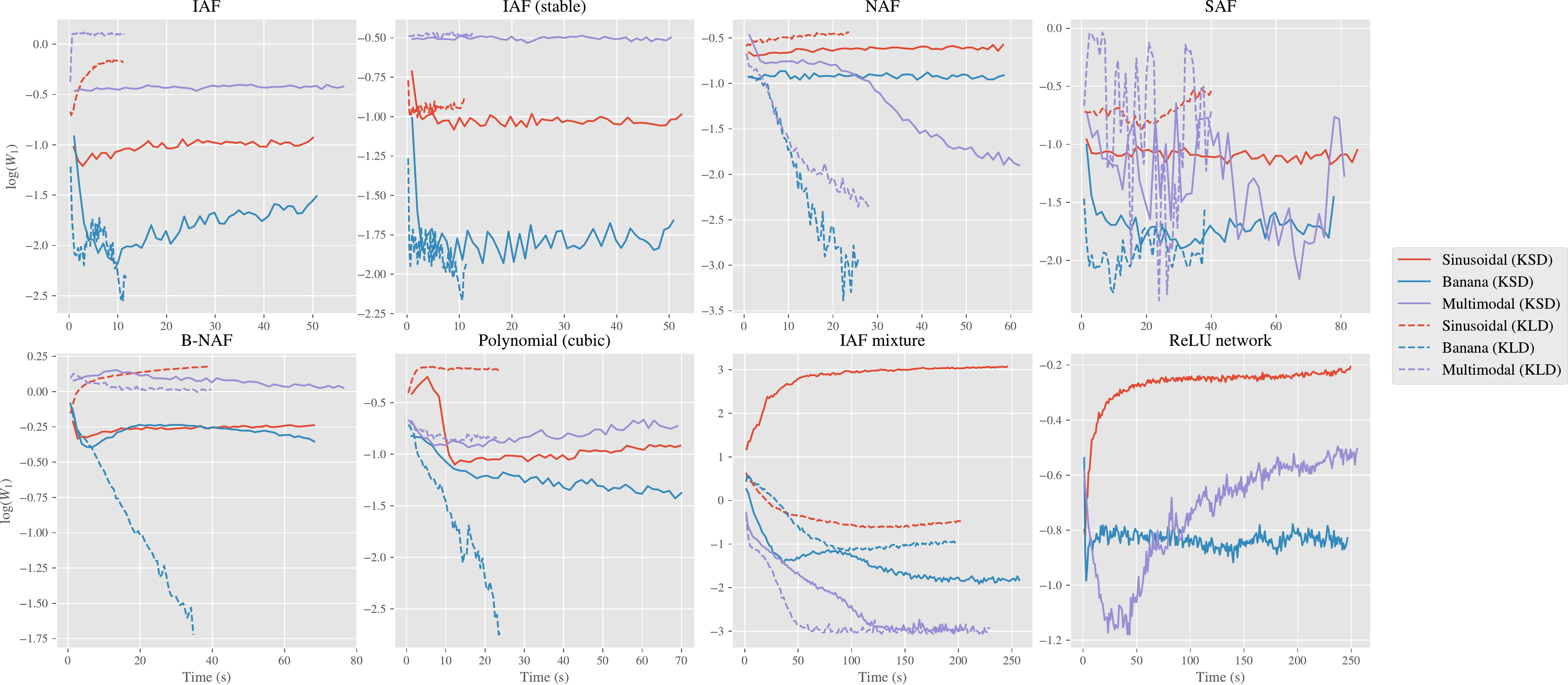}

\caption{
    Wasserstein-1 metric, $W_1$, as a function of the CPU wall-clock time for each synthetic test experiment.
}
\label{fig: wass vs time}
\end{figure}

We now discuss the implementations details of all the methods used in \Cref{subsec: toy examples}. 
In all our measure transport implementations, unless specified otherwise, we used existing implementations in Pyro \citep{bingham2018pyro}. 
Furthermore, the reference measure used for synthetic tests was the standard Gaussian on $\mathbb{R}^p$.

\paragraph{Hamiltonian Monte Carlo:}{We used an adaptive HMC algorithm in which the integrator step size was automatically adjusted using a dual-averaging algorithm in a warm-up phase to give an average acceptance statistic of 0.8 \citep{betancourt2014optimizing} and the number of integrator steps per transition was set dynamically by expanding the trajectory until a termination criterion was met \citep{hoffman2014no,betancourt2017conceptual}. We used the HMC implementations in the Python package \citep{graham2020mici}, with the dual-averaging adaptation algorithm settings following the defaults used in Stan \citep{carpenter2017stan}. Only the post-warm-up samples were included in estimates of the discrepancies and density plots.
}

In the following, the autoregressive neural networks that we specify are the Masked Autoencoders for Density Estimation (MADE) of \cite{germain2015made}; the only difference being, that there is no sigmoidal non-linearity applied to the output layer.

\paragraph{Inverse Autoregressive Flow (IAF):}{
 Recall from \Cref{subsec: transport maps section}, that an autoregressive flow $T:\mathbb{R}^d\rightarrow\mathbb{R}^d$ is of the form $T_i(x) = \tau(c_i(x_1,\ldots,x_{i-1}),x_i)$, where $T = (T_1,\ldots,T_d)$ and $x = (x_1,\ldots,x_d)$. The IAF flow of \cite{kingma} takes $c_i$ to be the $i$th output of an autoregressive neural network and $\tau$ to be an affine transform of the form
 \begin{equation*}
     \tau(c_i(x_1,\ldots,x_{i-1}), x_i) = \mu_i + \exp(s_i) x_i,
 \end{equation*}
 where $\mu_i\in\mathbb{R}$ and $s_i\in\mathbb{R}$ are outputs from $c_i(x_1,\ldots,x_{i-1})$. Note that the coefficient of $x_i$ is forced to be positive, this ensures the resulting transport map is monotonic. 

For each synthetic test problem, we used a single IAF where the dimensionality of the hidden units in the single hidden layer of the underlying autoregressive neural network was $40$. The underlying autoregressive neural network used the ReLU activation function. The IAF was initialised using the same default random initialisation in both the KSD and KLD experiments. $10,000$ iterations of Adam were used, with learning rate $0.001$.
}

\paragraph{Stable Inverse Autoregressive Flow (IAF stable):}{ 
Closely related to an IAF, a stable IAF was introduced in \cite{kingma} in order to improve numerical stability. The only difference being that the $\tau$ is of the form
 \begin{equation*}
     \tau(c_i(x_1,\ldots,x_{i-1}), x_i) = \text{sigmoid}(s_i)x_i + (1 - \text{sigmoid}(s_i))\mu_i,
 \end{equation*}
 where $\mu_i\in\mathbb{R}$ and $s_i\in\mathbb{R}$ are outputs from $c_i(x_1,\ldots,x_{i-1})$ and 
 \begin{equation*}
     \text{sigmoid}(x) = \frac{e^x}{1 + e^x},
 \end{equation*}
 where, for $x\in\mathbb{R}^d$, we consider $\text{sigmoid}$ to be applied component-wise. Since $\text{sigmoid}$ is monotonic, the resulting transport map is again monotonic. The restriction $\text{sigmoid}(s_i)\in (0,1)$ may limit the expressibility of the transport map compared to standard IAF, but this at the expense of increased numerical stability.
 
 For each synthetic test problem, we used a single stable IAF where the dimensionality of the hidden units in the single hidden layer of the underlying autoregressive neural network was $40$. The underlying autoregressive neural network used the ReLU activation function. The stable IAF was initialised using the same default random initialisation in both the KSD and KLD experiments. $10,000$ iterations of Adam were used, with learning rate $0.001$.
}

\paragraph{Neural Autoregressive Flow (NAF):}{
A NAF, introduced in \cite{huang2018neural}, is again an autoregressive flow which, compared to the preceding IAF and stable IAF, offers greater flexibility and a universality guarantee. The autoregressive conditioner $c$ is again taken as an autoregressive neural network and $\tau:\mathbb{R}\rightarrow\mathbb{R}$ takes the form of a monotonic neural network whose weights and biases are the output of the conditioner $c$. Monotonicity of $\tau$ is guaranteed by using strictly positive weights and strictly monotonic activation functions. 

The particular implementation of $\tau$ that we used in Pyro uses what is termed a \textit{deep sigmoidal flow} (DSF) in \cite{huang2018neural}. A DSF is a single layer dense neural network with a sigmoidal activation function. Furthermore, in order to increase the effective range of $\tau$, an inverse sigmoid function is taken on the output layer. However, this inverse sigmoid function has domain $(0,1)$ and thus the weights and biases of the output layer must be constrained such that this composition can be well defined. In a DSF, this is achieved by having no bias term on the output layer and constraining the output layer's weights, $w_{ij}$, to satisfy $\sum_i w_{ij} = 1$. Thus the output term of a DSF is a convex combination of the output of the hidden layer. The overall transformation of a DSF is thus of the form
\begin{equation*}
    f_\text{DSF} = \text{sigmoid}^{-1} \circ C \circ \text{sigmoid} \circ F 
\end{equation*}
where $F:\mathbb{R}\rightarrow\mathbb{R}^k$ is an affine transformation and $C:\mathbb{R}^k\rightarrow\mathbb{R}$ is a convex combination. $10,000$ iterations of Adam were used, with learning rate $0.001$.

Following \cite{huang2018neural}, for each synthetic test problem, we used a single DSF style NAF, where the dimensionality of the hidden sigmoid units in each DSF was $16$ and the dimensionality of the hidden units in the single hidden layer of the underlying autoregressive neural network was $40$. The underlying autoregressive neural network used the ReLU activation function. Note that the dimensionality of $i$th output $c_i$ of this autoregressive neural network is $48$, due to the $2\times 16$ weight terms and the $16$ bias terms in each DSF. The NAF was initialised using the same default random initialisation in both the KSD and KLD experiments.
}

\paragraph{Spline Autoregressive Flow (SAF):}{
A SAF, developed in \cite{durkan2019neural} and \cite{dolatabadi2020invertible}, is again an autoregressive flow that takes $\tau:\mathbb{R}\rightarrow\mathbb{R}$ as a piecewise monotonic rational polynomial function (a spline) on an interval $[-a,a]$ and the identity otherwise. A rational polynomial function is the ratio of two polynomials. In \cite{durkan2019neural}, the polynomial was taken as quadratic polynomial and in \cite{dolatabadi2020invertible}, the polynomial was taken as a linear polynomial. The parameters controlling each rational polynomial function, are the output of an autoregressive neural network. We also note that originally the spline transform was implemented in the context of coupling flows, rather than autoregressive flows.

For each synthetic test problem, we used a single SAF based on rational linear splines with $8$ pieces defined on the interval $[-3,3]$. The underlying autoregressive neural network had two hidden layers, each of dimension $20$. The SAF was initialised using the same default random initialisation in both the KSD and KLD experiments. $10,000$ iterations of Adam were used, with learning rate $0.001$.
}

\paragraph{Block Neural Autoregressive Flow (B-NAF):}{
A B-NAF, introduced in \cite{decao2019block}, is similar in spirit to the NAF. It is an autoregressive flow, where $\tau$ is a neural network. 
The difference is now that the weights and biases of $\tau$ are not the output of an autoregressive conditioner network $c$; instead, the parameters of the neural network are learned directly. In a B-NAF, the affine transformations $L:\mathbb{R}^{ad} \rightarrow \mathbb{R}^{bd}$, for $a,b\in\mathbb{Z}^+$, used at a given layer are always in a lower triangular block form
\begin{equation*}
    L(x) = \begin{pmatrix}
    u(B_{11}) & 0 & \ldots & 0 \\
    B_{21} & u(B_{22}) & \ldots & 0 \\
    \vdots & \vdots & \ddots & \vdots \\
    B_{d1} & B_{d1} & \ldots & u(B_{dd})
    \end{pmatrix}x + \mu,
\end{equation*}
where $u:\mathbb{R}\rightarrow\mathbb{R}^+$, each $B_{ii}\in \mathbb{R}^{a\times b}$ and $\mu \in \mathbb{R}^{bd}$ is the freely parameterised bias term . The positivity-ensuring transform $u$ enforces monotonicity. Bijectivity is further ensured by using bijective activation functions. Note that this particular form of affine transformation place restrictions on the structure of the neural network. For instance, the hidden dimensions must be a multiple of the input dimension $d$. Similarly to NAFs, B-NAFs also have a universality result.

For each synthetic test problem we used a single B-NAF with the $\tanh$ activation function and $u(x) = \exp(x)$. The B-NAF used had two hidden layers and was of the form
\begin{equation*}
    f_{BNAF} = L_3 \circ \tanh \circ L_2 \circ \tanh \circ L_1,
\end{equation*}
with lower triangular block affine transformations $L_1:\mathbb{R}^2\rightarrow\mathbb{R}^{16}$, $L_2:\mathbb{R}^{16}\rightarrow\mathbb{R}^{16}$ and $L_3:\mathbb{R}^{16}\rightarrow\mathbb{R}^2$. The B-NAF was initialised using the same default random initialisation in both the KSD and KLD experiments. $10,000$ iterations of Adam were used, with learning rate $0.001$.
}

\paragraph{Polynomial (Cubic):}{
 Polynomials were first put forward as possible parametric transport maps in measure transport \citep{Marzouk_2016, parno2018transport}. In \cite{Marzouk_2016}, each component of a polynomial transport map $T:\mathbb{R}^d\rightarrow\mathbb{R}^d$, was parameterised as a linear basis expansion of multivariate polynomials $\phi_j:\mathbb{R}^d\rightarrow\mathbb{R}$. Each $\phi_j$ is further parameterised with respect to a vector of polynomial degrees $j = (j_1,\ldots,j_d)\in\mathbb{N}^d$ as a product of $d$ univariate polynomials of the form
 \begin{equation*}
     \phi_j(x) = \prod_{k=1}^d \psi_{j_k}(x_i),
 \end{equation*}
 where each $\psi_{j_k}(x_i)$ is a univariate degree $j_k$ polynomial. These $\psi_{j_k}$ can come from orthogonal families of polynomials or simply be monomials. For instance, we could take the $\psi_{j_k}$ to be orthogonal with respect to the reference measure of the transport map. The $i$th component of $T$ can thus be written as
 \begin{equation*}
     T_i(x) = \sum_{j \in \mathcal{J}_i} \lambda_{j,i}\phi_j(x),
 \end{equation*}
 where each $\lambda_{j,i}\in\mathbb{R}$. This is a flexible parameterisation that can enforce triangularity through the choices of the $\mathcal{J}_i$. For example, a natural choice to enforce triangularity would be to take $\mathcal{J}_i$ to consist of vectors $j = (j_1,\ldots,j_i,0,\ldots,0)$ such that $\sum_{k}j_k \leq p$. The first constraint enforces triangularity and the second restraint ensures that the total degree of the resulting polynomials would be no greater than a given $p \in\mathbb{N}$. In higher dimensions, this may not be practical since the number of parameters grows quickly as the dimension increases. Thus other constraints on $\mathcal{J}_i$ were put forward, such as removing mixed terms in the basis. 
 
 An issue with this approach is that the resulting maps are not monotonic for all values of the coefficients $\lambda_{j,i}$. In \cite{Marzouk_2016} and \cite{parno2018transport}, monotonicity was constrained locally at a given set of samples $\{u_i\}_{i=1}^n$ from the reference distribution. Due to the triangular nature of the transport map, this effectively results in a finite set of linear constraints of the form $\partial_{x_i} T_i(u_k) > 0$ for $i = 1,\ldots d$ and $k = 1,\ldots, n$. In our implementation of KLD-based polynomial transport for the synthetic test bed, we found that this approach was not necessary since, in each case, the resulting Jacobian always had a positive determinant.
 
In our synthetic experiments, we used the the natural choice of the $\mathcal{J}_i$ that enforces triangularity that we previously discussed with $p=3$ and took the $\psi_i$ as simple monomials. The overall transport map was thus a multivariate cubic polynomial of the form
 \begin{equation*}
     \begin{pmatrix}
     T_1(x_1) \\
     T_2(x_1,x_2)
     \end{pmatrix} = \begin{pmatrix}
     \sum_{i=0}^3 c_i x_1^i \\
     \sum_{i=0}^3 \sum_{j=0}^{3-i} c_{ij} x_1^ix_2^{j}  \\
     \end{pmatrix}.
 \end{equation*}
 The polynomial transport map was initialised to the identity in all synthetic experiments. $10,000$ iterations of Adam were used, with learning rate $0.001$.
}

\paragraph{IAF mixture:}{
A mixture of transport maps is a distribution of the form
\begin{equation*}
    \sum_{i=1}^n w_i T_\#^{(i)} Q_i,
\end{equation*}
where the $T^{(i)}$ are each a transport map of a given form, the $Q_i$ are possibly distinct reference distributions and the mixing weights $w_i \geq 0$ satisfy $\sum_i w_i = 1$. This is a very flexible extension to using just a single transport map. Furthermore, in principle, both the number of mixing components $n$ and the mixing weights $w_i$ could be learnt. For example, the weights $w_i$ could be the output of a neural network with $\text{softmax}$ applied to the output layer\footnote{The $i$th component of the softmax function is of the form $\text{softmax}(x)_i = \frac{\exp(x_i)}{\sum_{j}\exp(x_j)}$.}, as was done in \cite{pires2020variational}.

For simplicity, in our synthetic experiments, we \textit{a priori} set $n = 4$ and further set each $w_i = 1/4$. We took each $T^{(i)}$ as a single IAF, where the dimensionality of the hidden units in the single hidden layer was $8$. Refer to our discussion of an IAF in \Cref{subsec: details of toy models} or \cite{kingma} for full details of an IAF. The $Q_i$ were initialised as Gaussians with means $(-2,2), (-2,-2), (2,-2), (2,2)$ respectively and each with identity covariance matrix. $30,000$ iterations of Adam were used, with learning rate $0.001$.
}

\paragraph{ReLU network:}{
Refer to \Cref{def: DNN} for the definition of a deep ReLU network. 

For our synthetic experiments, we implemented a deep ReLU network for each synthetic test problem. 
When using KSD, the transport map need not be a diffeomorpism and so, to illustrate this flexibility, the input dimension of the ReLU network for each experiment was taken as $4$ (while the dimension of the target was 2). For each problem, the ReLU network $f_{ReLU}:\mathbb{R}^4\rightarrow\mathbb{R}^2$ had two hidden layers and was of the form
\begin{equation*}
    f_{ReLU} =   F_3\circ \sigma \circ F_2 \circ \sigma \circ F_1,
\end{equation*}
where $F_1:\mathbb{R}^4\rightarrow\mathbb{R}^{20}$, $F_2:\mathbb{R}^{20}\rightarrow\mathbb{R}^{20}$ and $F_3:\mathbb{R}^{20}\rightarrow\mathbb{R}^2$ are affine transformations and $\sigma$ is the ReLU non-linearity, defined in \Cref{app: DNN}. 
Using the default random initialisation of these ReLU networks nearly always resulted in bad output. So, for each synthetic experiment, the ReLU network was pretrained for $10,000$ iterations of KSD-based measure transport using Adam with learning rate $0.001$, in order to approximate the reference distribution $\mathcal{N}((0,0),I_2)$. 
That is, we pretrained the ReLU network in order to initialise it close to $T_\# Q \approx \mathcal{N}((0,0),I_2)$.
After pretraining, $50,000$ further iterations of Adam were used for each synthetic test problem, with learning rate $0.001$.
}

\subsection{Details of the Biochemical Oxygen Model Experiment} \label{subsec: details of biochemical model} 

\paragraph{Derivation of the Posterior:}
Following on from \Cref{subsec: biochemical oxygen}, recall that the two-dimensional biochemical oxygen demand model is of the form
\begin{equation*}
    B(t) = \alpha_1(1 - \exp(-\alpha_2t)).
\end{equation*}
Due to the positivity constraints on $\alpha_1$ and $\alpha_2$, we perform inference on the log of the parameters and thus consider the model
\begin{equation*}
    B(t; \theta_1,\theta_2) = e^{\theta_1}(1 - \exp(-e^{\theta_2}t)).
\end{equation*}
Synthetic data $y = (y_i)_{i=1}^6$ were generated at times $t = 0,1,2,3,4,5$ with the parameter values $\theta_1 = \log(1)$ and $\theta_2 = \log(0.1)$, with observations corrupted by independent mean $0$ Gaussian errors with variance $\sigma^2 = 0.05^2$. See \Cref{fig: bod data} for a plot of $B(t;\theta_1,\theta_2)$ with these given parameter values alongside our generated synthetic data.

\begin{figure}[h!] 
   \centering
   \includegraphics[width=0.4\textwidth]{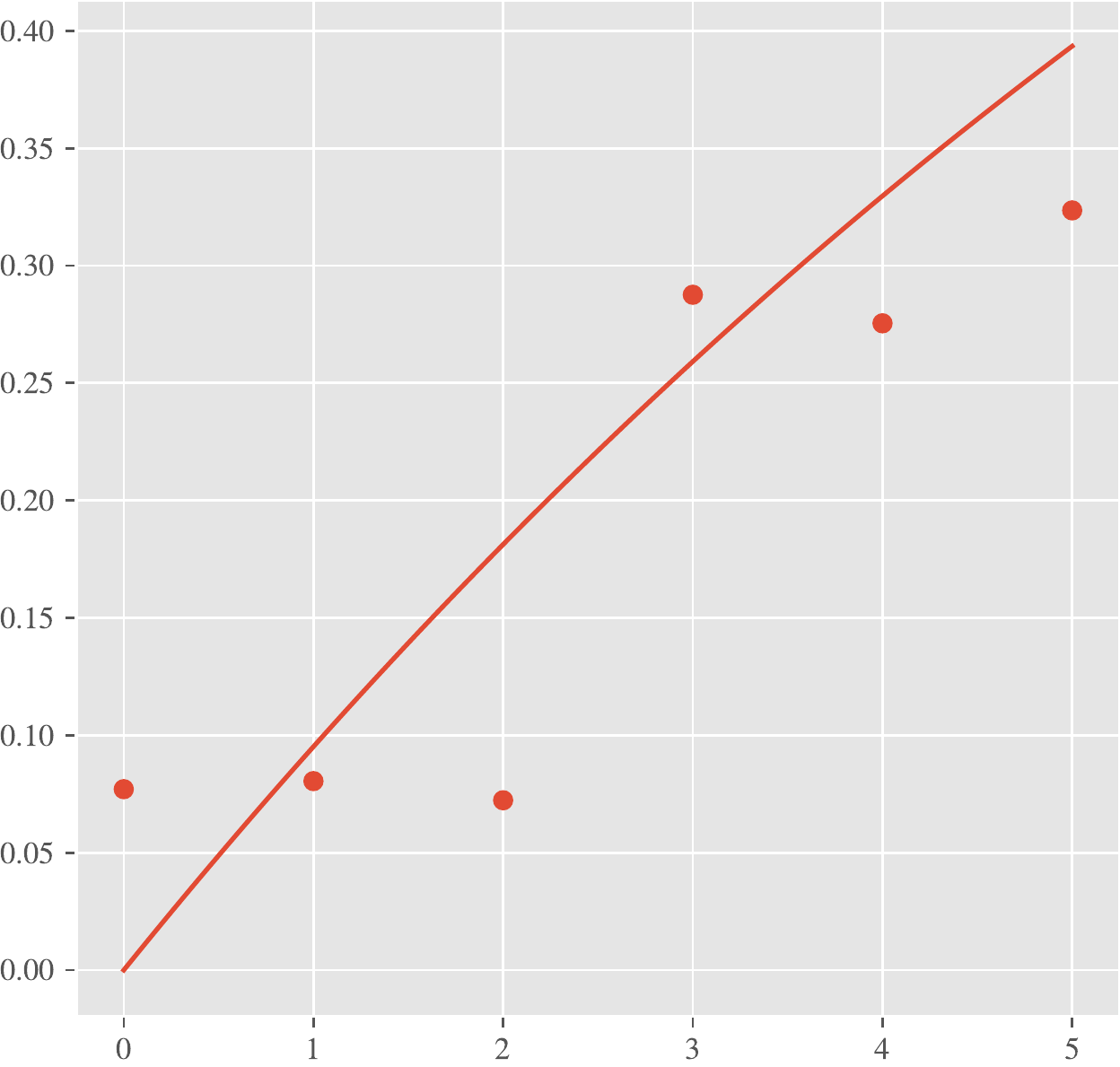}

\caption{
    Plot of $B(t; \log(1),\log(0.1))$ with the corresponding synthetic data at $t = 0,1,2,3,4,5$.
}
\label{fig: bod data}
\end{figure}
The likelihood is thus of the form
\begin{equation*}
    p(y\,|\,\theta_1,\theta_2) = \prod_{i=1}^6 \mathcal{N}(y_i; B(t_i; \theta_1,\theta_2),\sigma^2),
\end{equation*}
The prior specified for $\theta = (\theta_1,\theta_2)$ was $\theta\sim\mathcal{N}((0,0),I_2)$. The resulting posterior density is thus of the form
\begin{equation*}
    p(\theta_1,\theta_2\,|\,y) \propto \mathcal{N}(\theta_1,\theta_2;(0,0),I_2) \prod_{i=1}^6 \mathcal{N}(y_i; B(t_i; \theta_1,\theta_2),\sigma^2).
\end{equation*}
\paragraph{Methodology:}
Our choice of parametric transport map was a Block Neural Autoregressive Flow (B-NAF) of \cite{decao2019block}. We used the same B-NAF as the one used in \Cref{subsec: toy examples}, where we again used a B-NAF with two hidden layers of the form
\begin{equation*}
    f_{BNAF} =  L_3 \circ \tanh \circ L_2 \circ \tanh \circ L_1,
\end{equation*}
with lower triangular block affine transformations $L_1:\mathbb{R}^2\rightarrow\mathbb{R}^{16}$, $L_2:\mathbb{R}^{16}\rightarrow\mathbb{R}^{16}$ and $L_3:\mathbb{R}^{16}\rightarrow\mathbb{R}^{2}$. Refer to \Cref{subsec: details of toy models} or to \cite{decao2019block} for a full description of a B-NAF. The lengthscale used for KSD was $\ell = 0.1$. We again used the Adam optimiser, with default learning rate $0.001$ with $30,000$ iterations for each method.

\paragraph{Results:} See \Cref{fig: bod results} for samples obtained from each resulting transport map. The KSD-based method obtained a Wasserstein-1 distance of $0.069$ and the KLD-based method obtained a Wasserstein-1 distance of $0.015$.  Refer to \Cref{subsec: performance metrics} for details on how this was calculated. \Cref{fig: bod prior to posterior} plots $B(t;\theta_1,\theta_2)$ using samples from the prior and the approximate posterior using KSD-based measure transport.

\begin{figure}[h!] 
   \centering
   \includegraphics[width=0.8\textwidth]{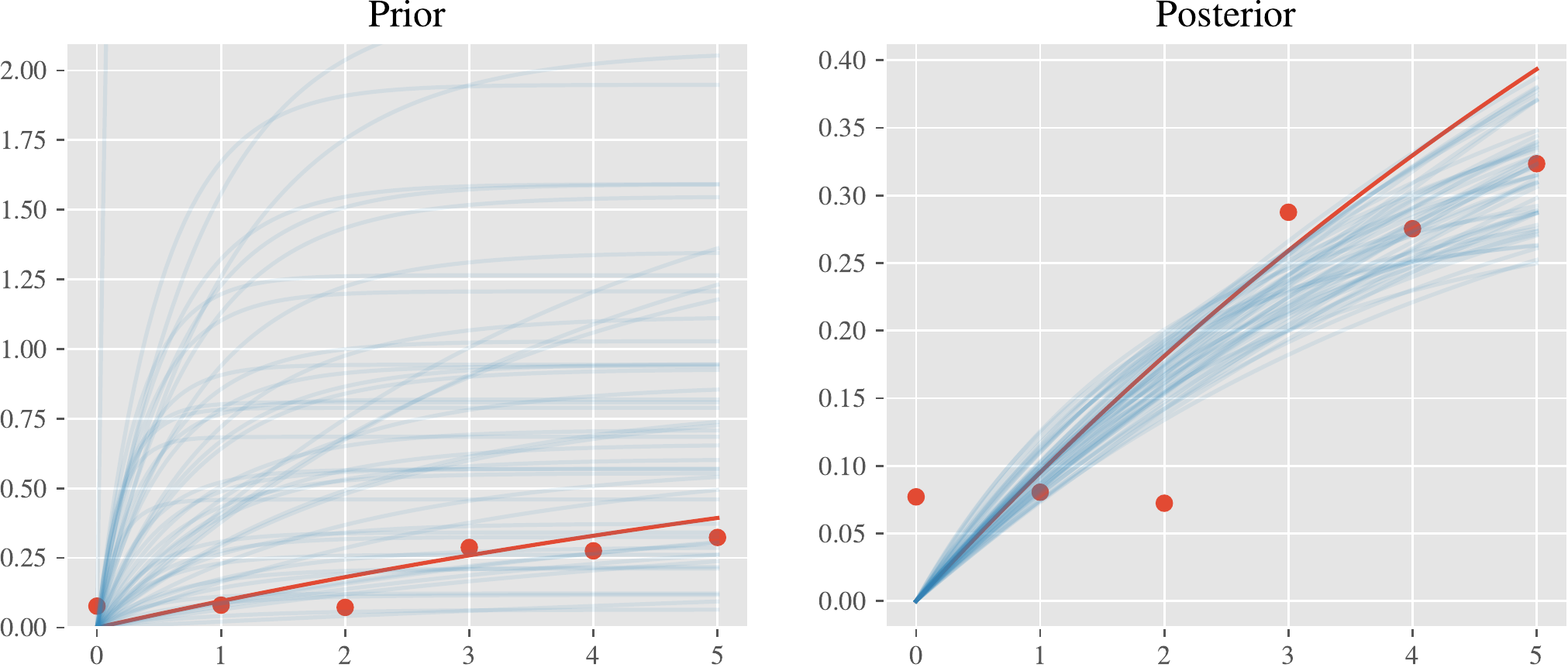}
\caption{
    Plot of $B(t;\theta_1,\theta_2)$ using $50$ samples for $\theta$ from (left) the prior and (right) the posterior, as approximated using KSD-based measure transport. The red line is $B(t;\log(1),\log(0.1))$.
}
\label{fig: bod prior to posterior}
\end{figure}

\subsection{Details of the Generalised Lotka--Volterra Model Experiment}  \label{subsec: details of lotka volterra}

\paragraph{Prior Specification:}

Recall that, from \Cref{subsec: lotka-volterra}, the generalised Lotka--Volterra model we considered was of the form
\begin{align*} 
    \frac{\mathrm{d}p}{\mathrm{d}t}(t) &=  rp(t)\left(1 - \frac{p(t)}{k}\right) - s\frac{p(t)q(t)}{a + p(t)} \\
    \frac{\mathrm{d}q}{\mathrm{d}t}(t) &= u\frac{p(t)q(t)}{a + p(t)} - vq(t),
\end{align*}
with parameters $p_0,q_0,r,k,s,a,u,v > 0$.
These parameters are physical quantities, see \cite{rockwood2015ecology} for their full meaning. 
Due to the positivity constraints on these parameter values, similar to our biochemical oxygen demand experiment \Cref{subsec: biochemical oxygen}, we again perform inference on the log of the parameters. We thus consider the model
\begin{align*} 
    \frac{\mathrm{d}p}{\mathrm{d}t}(t) &=  e^{r'} p(t)\left(1 - \frac{p(t)}{e^{k'}}\right) - e^{s'}\frac{p(t)q(t)}{e^{a'} + p(t)} \\
    \frac{\mathrm{d}q}{\mathrm{d}t}(t) &= e^{u'}\frac{p(t)q(t)}{e^{a'} + p(t)} - e^{v'}q(t),
\end{align*}
where we perform inference on the parameter $\theta = (p_0',q_0',r',k',s',a',u',v') \in \mathbb{R}^8$. 

After an investigation of the sensitivities of the solutions of the ODE with respect to the parameter values, we specified the following independent prior 
\begin{alignat*}{2}
    p_0' &\sim \mathcal{N}(\log 45, 0.2^2), q_0' &&\sim \mathcal{N}(\log 7, 0.3^2) \\
    r' &\sim \mathcal{N}(\log 0.5, 0.3^2), k' &&\sim \mathcal{N}(\log 80, 0.15^2) \\
    s' &\sim \mathcal{N}(\log 1.3, 0.2^2), a' &&\sim \mathcal{N}(\log 30, 0.1^2) \\
    u' &\sim \mathcal{N}(\log 0.6, 0.1^2), v' &&\sim \mathcal{N}(\log 0.28, 0.07^2).
\end{alignat*}
Letting $\mu \in \mathbb{R}^8$ be the vector of these given mean values and $K$ the diagonal matrix with these given variances on the diagonal, we have $\theta \sim \mathcal{N}(\mu,K)$. This specification results in log-normal priors on the exponentiated parameters.

Synthetic data $y = (p_i,q_i)_{i=1}^6$ were generated at times $t = 0,10,20,30,40,50$ with parameter values $(p_0',q_0',r',k',s',a',u',v') = (\log 50,\log 5, \log 0.6,\log 90, \log 1.2,\log 25, \log 0.5, \log 0.3)$; these data were perturbed by independent mean $0$ Gaussian errors with variance $\sigma^2 = 40$. 
See \Cref{fig: pred prey data} for a plot of the solution of the generalised Lotka-Volterra model with these given parameters alongside our generated synthetic data. 

\begin{figure}[h!] 
   \centering
   \includegraphics[width=0.4\textwidth]{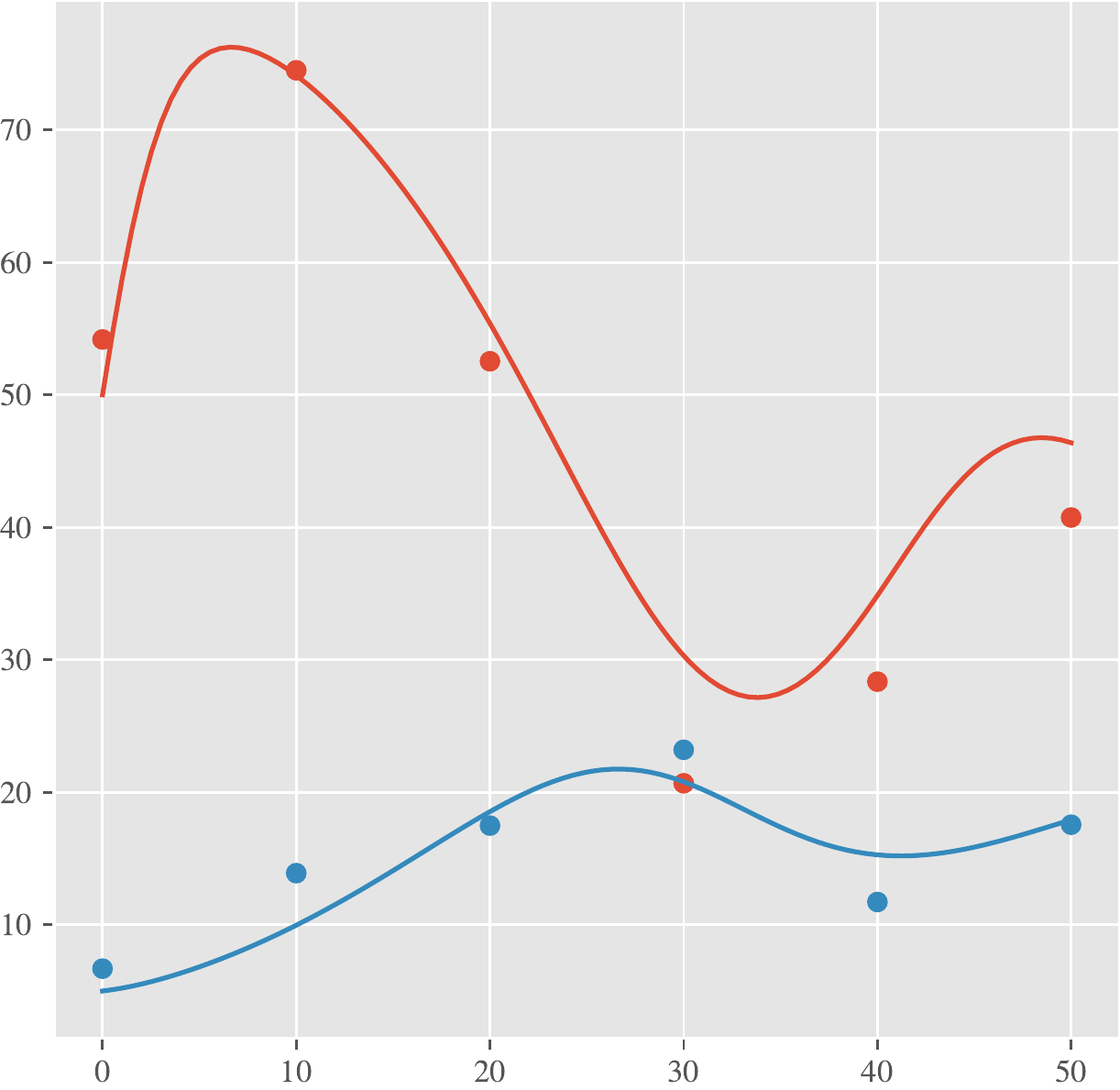}
\caption{
    Solution of the generalised Lotka--Volterra model and the data that were provided at $t = 0,10,20,30,40,50$.
}
\label{fig: pred prey data}
\end{figure}

The likelihood is thus of the form
\begin{equation*}
    p(y\,|\,\theta) = \prod_{i=1}^6 \mathcal{N}(p_i,q_i; (p_\theta(t_i),q_\theta(t_i)), \sigma^2 I_2),
\end{equation*}
where $p_\theta(t),q_\theta(t)$ are the solutions of generalised Lotka-Volterra ODE \eqref{eq: LV} with given parameter $\theta$. The resulting posterior density is thus of the form
\begin{equation*}
    p(\theta\,|\,y) \propto \mathcal{N}(\theta;\mu,K)\prod_{i=1}^6 \mathcal{N}(p_i,q_i;(p_\theta(t_i),q_\theta(t_i)), \sigma^2 I_2).
\end{equation*}

\paragraph{Methodology:} We used the \texttt{torchdiffeq} Python library \citep{chen2018neural} in order to numerically solve the Lotka--Volterra model and further utilised Pytorch's automatic differentiation capabilities to propagate gradients through the solver. In our implementation, we used the default Dormand-Prince Runge-Kutta method. The reference measure used was the prior. Our choice of parametric transport map was a B-NAF with $u(x) = \exp(x)$, of the form
\begin{equation*}
    f_{BNAF} =  L_3 \circ \tanh \circ L_2 \circ \tanh \circ L_1,
\end{equation*}
with lower triangular block affine transformations $L_1:\mathbb{R}^8\rightarrow\mathbb{R}^{64}$, $L_2:\mathbb{R}^{64}\rightarrow\mathbb{R}^{64}$ and $L_3:\mathbb{R}^{64}\rightarrow\mathbb{R}^{8}$. Refer to \Cref{subsec: details of toy models} or to \cite{decao2019block} for a full description of a B-NAF.

For both the KSD and KLD experiments, we used the same random initialisation of the B-NAF and pretrained on $10,000$ iterations of KLD-based measure transport on the prior (the reference measure), to ensure that the initial pushforward of samples through the B-NAF resulted in non-degenerate solutions of the Lotka--Volterra model.

The lengthscale used for KSD was $\ell = 0.1$ and we again used the Adam optimiser, with default learning rate $0.001$ with $50,000$ iterations for each method.

\paragraph{Results:} The KSD-based method obtained a Wasserstein-1 distance of $0.130$, whereas the KLD-based method acheived a Wasserstein-1 distance of $0.110$. Refer to \Cref{subsec: performance metrics} for details on how this was calculated. The resulting approximating distributions for both the KSD and KLD methods are plotted in \Cref{fig: lotka volterra marginals}.

\begin{figure*}[h!] 
    \centering
    \begin{subfigure}[b]{0.495\textwidth}
        \centering
        \includegraphics[width = \textwidth]{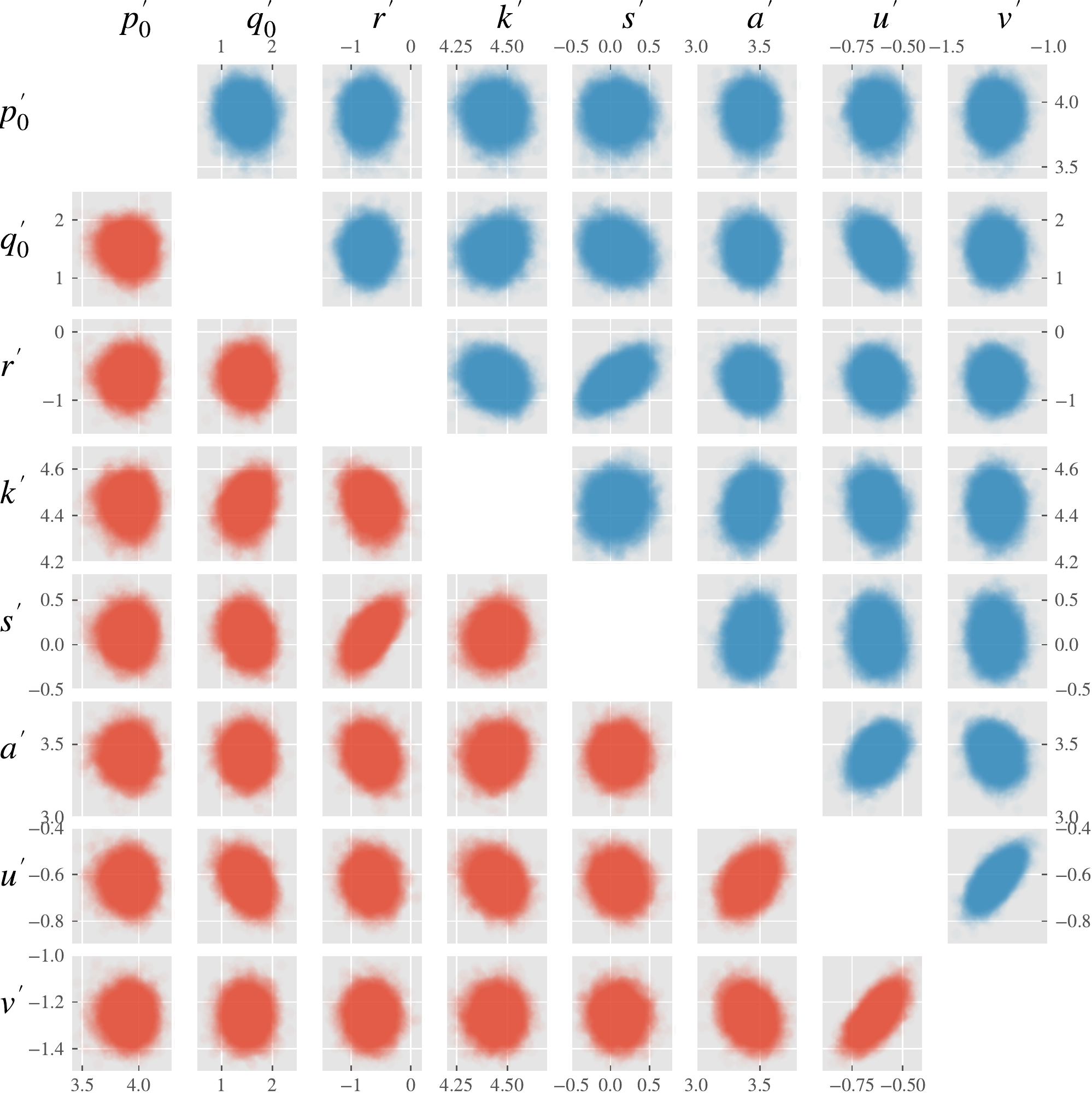}
        \caption{KSD}
        \label{subfig: LV KSD}
    \end{subfigure}%
    \hfill
    \begin{subfigure}[b]{0.495\textwidth}
        \centering
        \includegraphics[width = \textwidth]{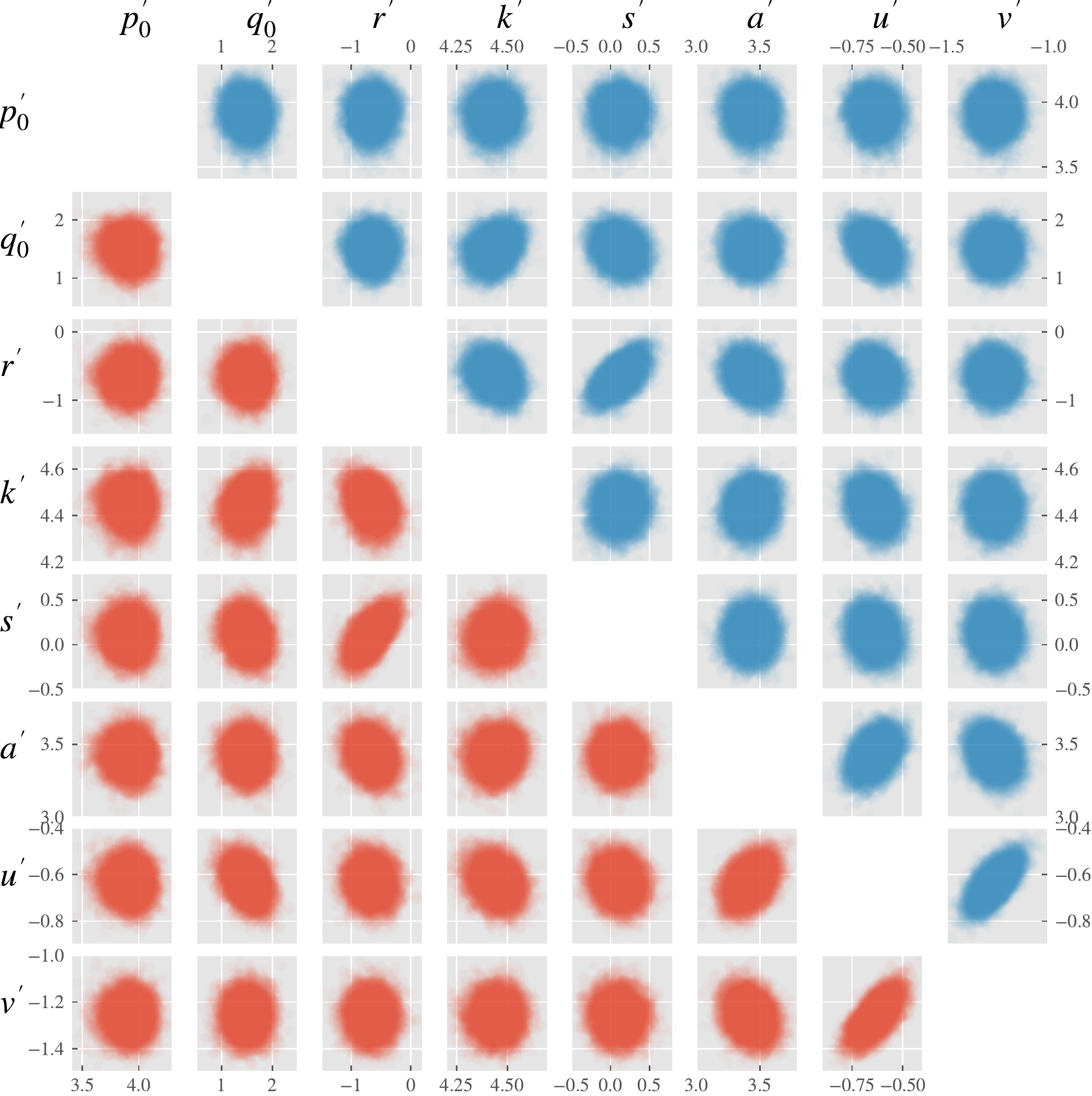}
        \caption{KLD}
        \label{subfig: LV KLD}
    \end{subfigure} 
    \caption{Two-dimensional projections of samples from (a) the KSD-based method and (b) the KLD-based method. In both plots, the lower triangular subplots are the two-dimensional projections of samples from the gold-standard HMC and the upper triangular subplots are the two-dimensional projections of samples from the two methods.}
    \label{fig: lotka volterra marginals}
\end{figure*}

\section{Further Investigations} \label{sec: investigations}

Here we report a series of further investigations, that explore specific aspects of KSD-based measure transport in more detail.

\subsection{Initialisation of Parameters in the Transport Map} \label{subsec: initialise}
Both KLD and KSD-based measure transport can be sensitive to the initialisation of the parameters in a given transport map. 
This is, for instance, evidenced in \Cref{fig: multimodal failure}. In our experiments we generally used a random initialisation as specified by their implementations in Pyro \citep{bingham2018pyro}. In \Cref{sec: computational details}, we specify for each experiment in the main paper what initialisation was used and whether we pretrained on the reference distribution $Q$. 

A general remedy for poor initialisation is either to pretrain on the reference distribution. This was done in our applied examples in \Cref{subsec: biochemical oxygen} and \Cref{subsec: lotka-volterra}. 
Alternatively, in a Bayesian inference context one could pretrain on the prior distribution instead.
The latter approach may be advantageous since we are guaranteed that the target's support is contained within the prior's support.

\subsection{Investigating the Choice of Stochastic Optimisation Method} \label{subsec: stochastic optimisation}
In all our experiments in \Cref{sec:experiments}, we used the Adam optimiser of \cite{kingma2014adam} with a fixed batch size of $100$ and with a varying number of iterations. In this section we explore how the output of KSD-based measure transport, with a fixed number of iterations of stochastic optimisation, interacts with the batch size as well as the stochastic optimisation method used. We fixed the target distribution as the multimodal problem and considered only B-NAF as our transport map with standard Gaussian reference distribution.  Results are shown in \Cref{table: other optimisation results}, where we report the Wasserstein-1 distance using $10^4$ samples (see \Cref{subsec: performance metrics} for more details). We pretrained the B-NAF on the reference distribution using KLD as our loss for $5000$ iterations of Adam with learning rate $l = 0.001$, hence the discrepancy with the main results reported in \Cref{table: toy density results}. 

\begin{table*}[h!] 
    \newcolumntype{Y}{>{\centering\arraybackslash}X}

\begin{tabularx}{\textwidth}{@{}lYYYY@{}} \toprule
          & \multicolumn{4}{c}{Batch size}  \\
      \cmidrule(lr){2-5}
     {Optimisation method} & {$25$} & {$50$} & {$100$} & {$200$} \\ \midrule
     Adam ($l = 0.001$)          & $0.594$  & $0.546$ & $0.121$ & $\bf{0.0827}$ \\
     Adam ($l = 0.01$)       & $0.663$ & $0.768$ & $0.726$  & $\bf{0.242}$ \\ \midrule
     Adagrad ($l = 0.001$)              & $0.630$ & $0.620$ & $0.617$  & $\bf{0.616}$ \\ 
     Adagrad ($l = 0.01$)                  & $0.612$ & $0.608$ & $0.606$  & $0.602$ \\  \midrule
     RMSprop ($l = 0.001$)              & $0.590$ & $0.504$ & $0.144$  & $\bf{0.0856}$ \\ 
     RMSprop ($l = 0.01$)                  & $1.00$ & $0.699$ & $\bf{0.0653}$ & $0.125$ \\  \midrule
     SGD ($l = 0.001$)              & $0.902$ & $\bf{0.194}$ & $0.402$  & $0.433$ \\ 
     SGD ($l = 0.005$)                  & $1.24$ & $0.265$ & $\bf{0.143}$  & $0.241$  \\  \midrule
     ASGD ($l = 0.001$)              & $0.701$ & $\bf{0.34}$ & $0.401$  & $0.432$ \\ 
     ASGD ($l = 0.005$)                  & N/A & $\bf{0.166}$ & $0.478$  & $0.391$\\  \bottomrule
\end{tabularx}

   \caption{The $W_1$ distance to the multimodal target, varying the stochastic optimisation method. We used $10,000$ iterations for each experiment. The $l$ value next to the name of each optimisation method was the learning rate used. Bold values indicate which of the batch sizes obtained the best Wasserstein-1 distance for each given optimisation method. N/A values indicate when the optimisation method failed. 
   }
\label{table: other optimisation results}
\end{table*}

From \Cref{table: other optimisation results}, the most consistent and best performing optimisation methods were Adam and RMSprop, where the smaller learning rate of $l = 0.001$ seemed to perform best. For stochastic gradient descent (SGD) and averaged stochastic gradient descent (ASGD), since the learning rate is non-adaptive, if the initial learning rate $l$ is too large, the optimiser can fail to converge. This is evidenced by ASGD failing at a batch size of $25$ with $l = 0.005$ in \Cref{table: other optimisation results}. For each of the optimisation methods detailed in \Cref{table: other optimisation results}, all parameters other than the learning rate were set to their default values as specified in Pytorch.

\subsection{Investigating the Effect of Quasi-Monte Carlo Sampling} \label{subsec: QMC investigation}

To reduce the variance of the Monte-Carlo based gradient estimators, it was put forward in \cite{alex2018quasimonte} and \cite{wenzel2018quasi}, to instead use randomised Quasi-Monte Carlo (QMC) in constructing an unbiased estimator of the gradient. 
This is achieved by simply replacing the Monte-Carlo samples from the base distribution with samples from a (randomised) QMC sequence in a principled manner. This may be especially useful in our setting, where the variance of the U-statistic estimator of KSD is often quite large. In this section, we explore the replacement of the Monte-Carlo based U-statistic estimator in \Cref{prop: gradients} with a QMC-based estimator empirically. 
We first, however, briefly outline the rudimentary idea. Refer to \cite{alex2018quasimonte} and \cite{wenzel2018quasi} for the full detail.

A low-discrepancy sequence or a QMC sequence of a given length on $[0,1]^d$, roughly speaking, allocates points such that the number of points in a given measurable subset of $[0,1]^d$ is proportional to its volume. 
A prototypical example of a randomised QMC estimator is the \textit{random shift modulo 1}, where the sequence is generated by first specifying a grid of values $x_i$ over $[0,1]^d$ and then sampling a $u\sim U([0,1]^d)$, the resulting QMC sequence is the set of points $x_i + u\mod{1}$. 
Using an appropriate measurable function $S:[0,1]^d\rightarrow\mathbb{R}^d$, one can use QMC to integrate with respect to a given distribution $P$, as long as $S_\#U([0,1]^d) = P$, by pushing forward the QMC sequence through $S$. 
\Cref{fig: QMC samples} plots a randomly shifted (modulo 1) grid in two-dimensions, along with its pushforward on to the standard Gaussian against a uniform sample.

Results are shown in \Cref{table: QMC results}, where we compare this prototypical QMC method with Monte Carlo; the convergence is shown in \Cref{fig: QMC vs MC convergence}. 
The transport map chosen was a NAF and of the same form as the NAF used in \Cref{subsec: toy examples} and specified in \Cref{subsec: details of toy models}. We used $8000$ iterations of Adam with learning rate $0.001$.
It appears that this QMC sequence generally performed worse than standard Monte-Carlo. 
These negative findings dissuaded us from further exploring QMC in this work.
However, it remains to be seen whether more advanced randomised QMC sequences, such as the scrambled Sobol sequence that was used in \cite{wenzel2018quasi}, provide performance gains relative to standard Monte Carlo.


\begin{figure*}[h!] 
    \centering
    \begin{subfigure}[b]{0.48\textwidth}
        \centering
        \includegraphics[width = \textwidth]{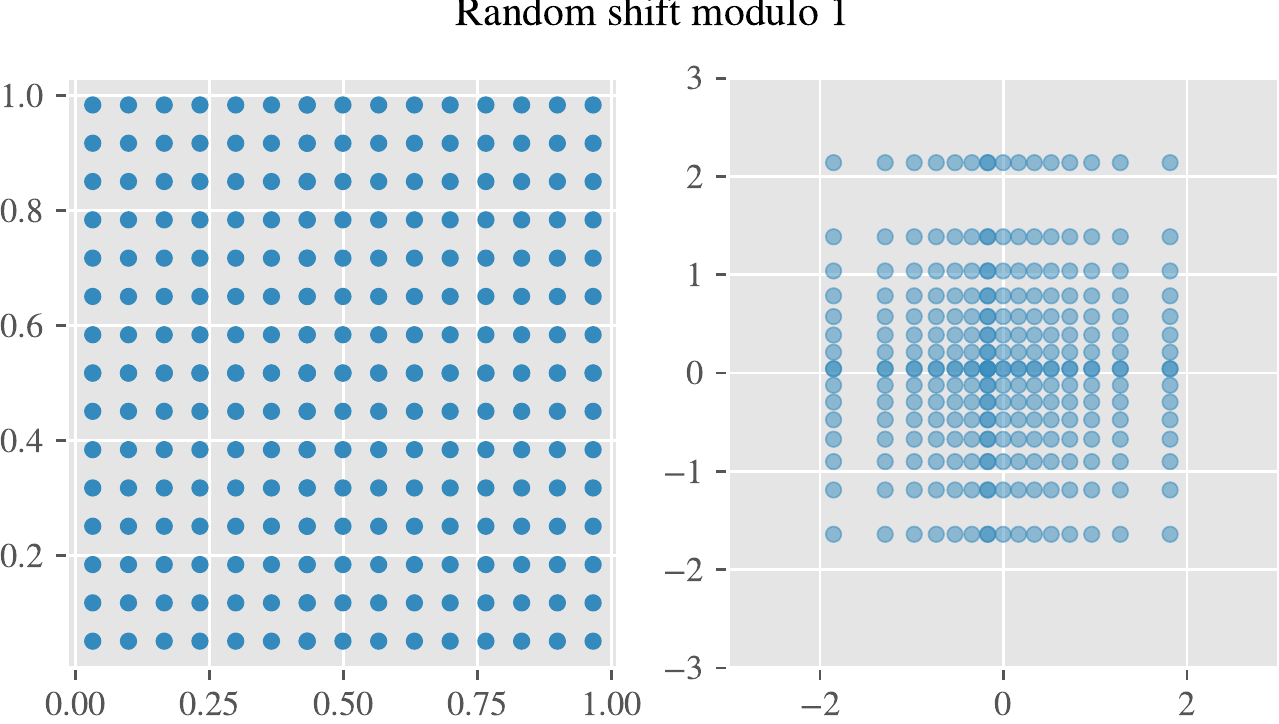}
        \caption{(Prototypical) quasi Monte Carlo}
    \end{subfigure}%
    \hfill
    \begin{subfigure}[b]{0.48\textwidth}
        \centering
        \includegraphics[width = \textwidth]{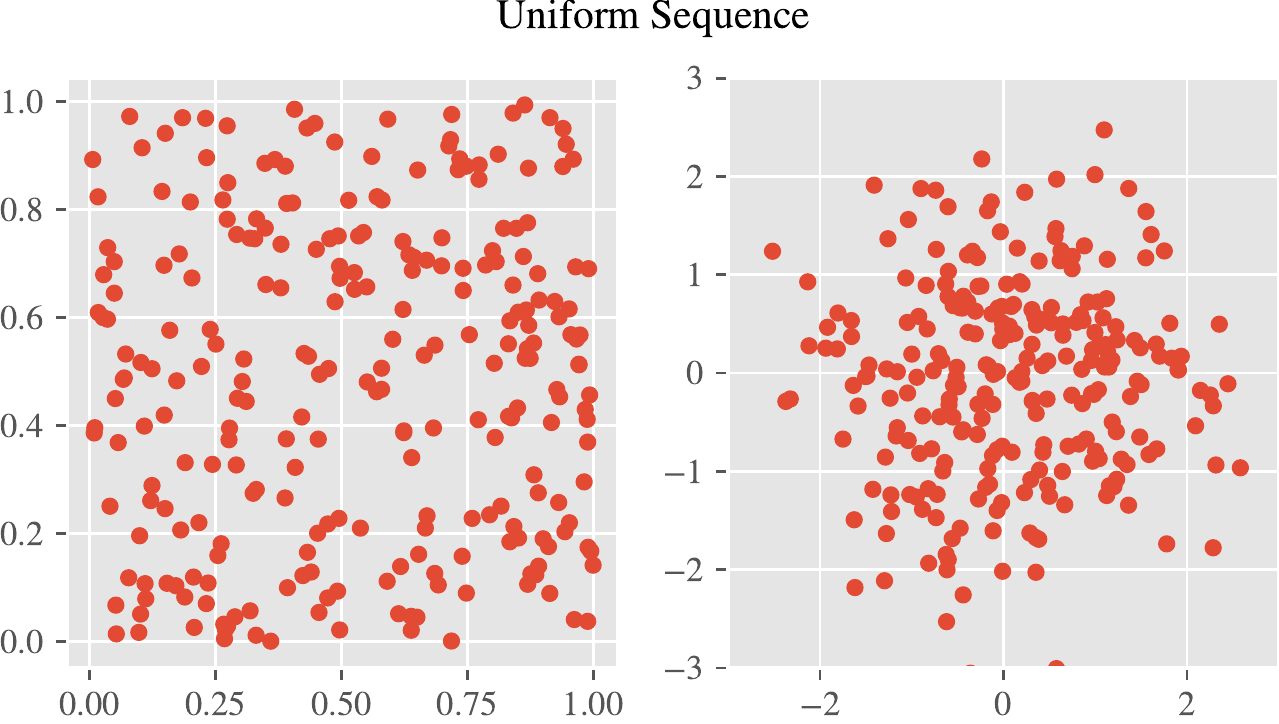}
        \caption{Monte Carlo}
    \end{subfigure} 
    \caption{A size $256$ quasi Monte Carlo point set (left) against a $256$ length uniform sequence (right), each alongside their pushforwards to the standard Gaussian.}
    \label{fig: QMC samples}
\end{figure*}


\begin{table*}[h!] 
    \newcolumntype{Y}{>{\centering\arraybackslash}X}

\begin{tabularx}{\textwidth}{@{}lYYY@{}} \toprule
     {Sampling Method} & {Sinusoidal} & {Banana} & {Multimodal}  \\ \midrule
     Random shifted (modulo 1) grid       & $0.59$  & $0.53$ & $0.22$  \\ 
     Monte Carlo      & $\bf{0.55}$  & $\bf{0.39}$ & $\bf{0.16}$  \\ \bottomrule
\end{tabularx}

   \caption{The resulting $W_1$ metrics from sampling using either quasi Monte Carlo or standard Monte Carlo. The transport map used was a NAF with $8000$ iterations of Adam. Bold values indicate which sampling rule performed best on each target.
   }
\label{table: QMC results}
\end{table*}

\begin{figure}[h!] 
   \centering
   \includegraphics[width=1\textwidth]{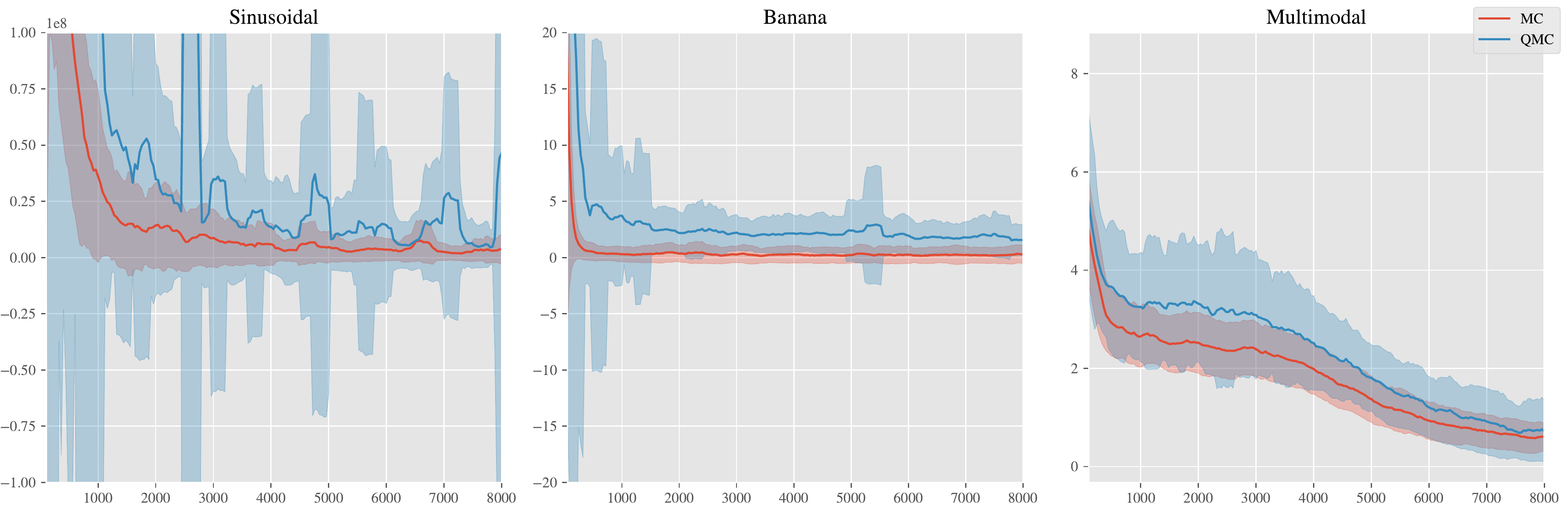}
\caption{
    Loss function against number of iterations of Adam for standard Monte Carlo (MC) and quasi Monte Carlo (QMC).
}
\label{fig: QMC vs MC convergence}
\end{figure}

\subsection{Investigating the Choice of Reference Distribution} \label{subsec: other Q}

All the experiments in the main text used a Gaussian distribution as the reference distribution $Q$.
However, different reference distributions could potentially offer some advantages, for instance in capturing thicker tails or multimodality \citep{izmailov2019semisupervised}. 
To investigate, we used the IAF as our transport map and compared the following reference distributions: a mixture of two Gaussians, a symmetric multivariate Laplace distribution, and the standard Gaussian used in \Cref{sec:experiments}. 
The mixture of two Gaussians reference distribution was of the form $\frac{1}{2}\mathcal{N}((0,-3),I_2) + \frac{1}{2}\mathcal{N}((0,3),I_2)$. The multivariate Laplace reference distribution was of the form $\text{Laplace}((0,0),I_2)$. The IAF we employed was the same one used in \Cref{subsec: toy examples} and fully specified in \Cref{subsec: details of toy models}. For each experiment we used the Adam optimiser with learning $0.001$ with $10,000$ iterations of Adam.
The target distributions were the synthetic distributions used in \Cref{subsec: toy examples} and fully specified in \Cref{subsec: details of toy models}. Results are shown in \Cref{table: different reference results} and notable output is shown in \Cref{fig: different reference}. 


\begin{table*}[h!] 
    \newcolumntype{Y}{>{\centering\arraybackslash}X}

\begin{tabularx}{\textwidth}{@{}lYYY@{}} \toprule
     {Reference Distribution $Q$} & {Sinusoidal} & {Banana} & {Multimodal}  \\ \midrule
     Laplace         & $0.45$  & $0.25$ & $1.4$  \\ 
     Gaussian Mixture          & \textbf{$0.23$}  & $0.25$ & \textbf{$0.46$}  \\
     Gaussian & $0.38$ & \textbf{$0.20$} & $0.67$ \\ \bottomrule
\end{tabularx}

   \caption{The resulting $W_1$ metrics from using different reference distributions $Q$ in KSD-based measure transport for the targets in \Cref{subsec: toy examples}. The transport map used for each experiment was an IAF and was optimised with $10,000$ iterations of Adam.
   }
\label{table: different reference results}
\end{table*}

\begin{figure}[h!] 
   \centering
   \includegraphics[width=1\textwidth]{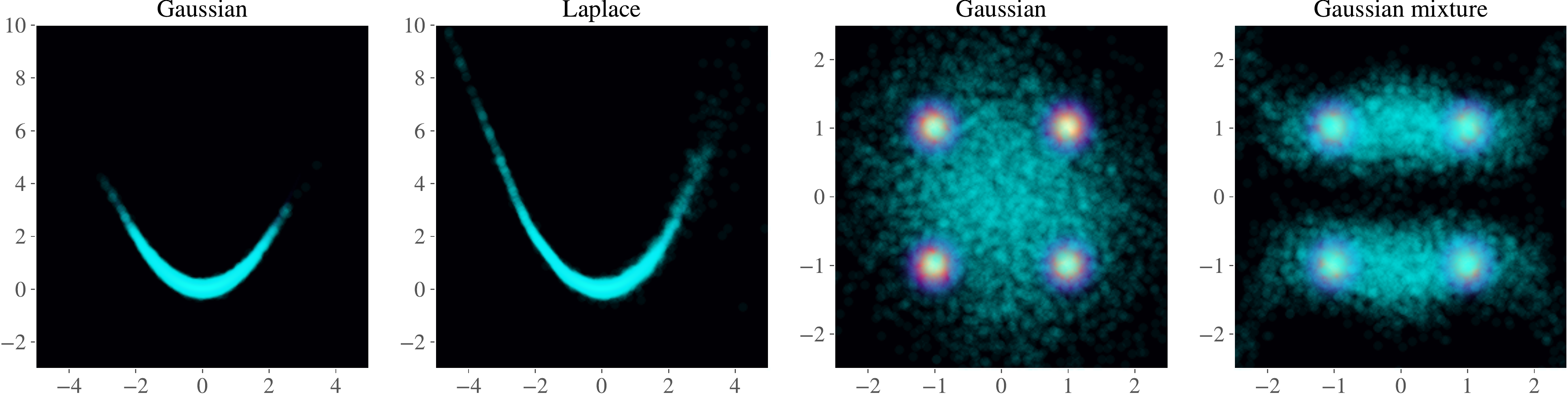}
\caption{
    Output of KSD-based measure transport using different reference distributions on the banana and multimodal targets. The reference distribution used in each experiment is indicated in the titles of the corresponding subplots.
}
\label{fig: different reference}
\end{figure}

Looking at \Cref{fig: different reference}, the heavier tails of the Laplace reference distribution resulted in heavier tailed output. Furthermore, the Gaussian mixture reference allowed the IAF two capture two modes, however it was unable to capture all four modes of the multimodal target.
Results in \Cref{table: different reference results} suggest it may be useful to consider the choice of $Q$ as part of the optimisation problem to be solved, although we did not attempt to do so in this work.

\subsection{Investigating the Choice of Lengthscale} \label{subsec: misspecified lengthscale}

In this section, we investigate how the choice of lengthscale $\ell$ of the inverse multi-quadric kernel (see \Cref{theorem:gorham inverse multi quadric}) can affect the output of KSD-based measure transport. We will see that, relative to the target distribution, if $\ell$ is too small the resulting output can be too focused if the target distribution has a relatively large dispersion, and on the other hand, if $\ell$ is too large, the resulting output can exhibit pathologies. To demonstrate, we will focus on the NAF transport map and consider simple Gaussian targets with covariance matrices $10I_2$ and $0.1I_2$. We will demonstrate the pathologies also occur in a more complex example of the multimodal problem encountered in \Cref{subsec: toy examples}. Output is shown in \Cref{fig: effect of lengthscale}. 

The NAF used is the same one used in the experiments in \Cref{subsec: toy examples} and specified in \Cref{subsec: details of toy models}. We used the Adam optimiser for $10,000$ iterations with learning rate $0.001$.

\begin{figure*}[h!] 
    \centering
    \begin{subfigure}[b]{0.33\textwidth}
        \centering
        \includegraphics[width = \textwidth]{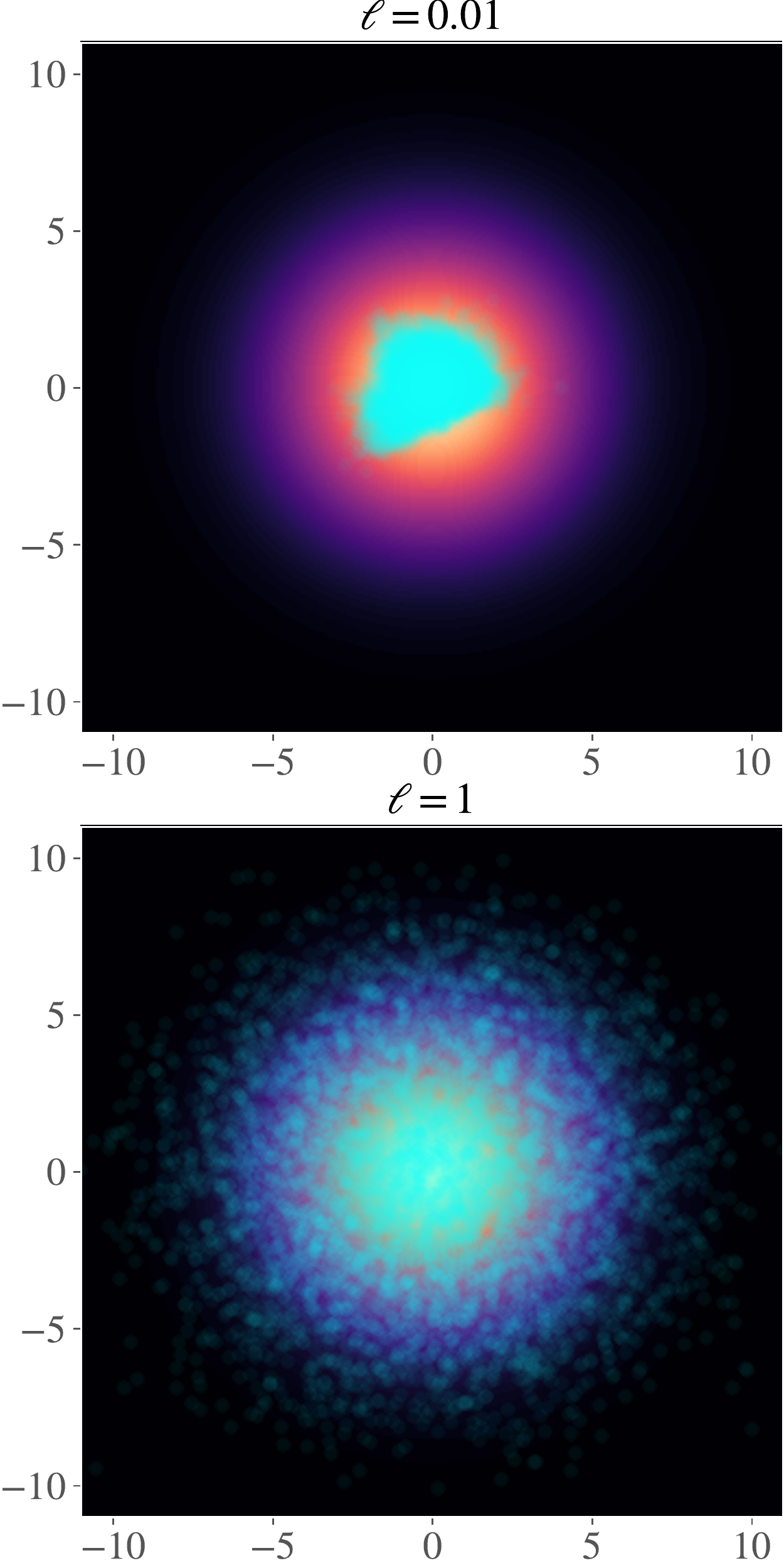}
        \caption{Gaussian $10I_2$}
    \end{subfigure}%
    \hfill
    \begin{subfigure}[b]{0.33\textwidth}
        \centering
        \includegraphics[width = \textwidth]{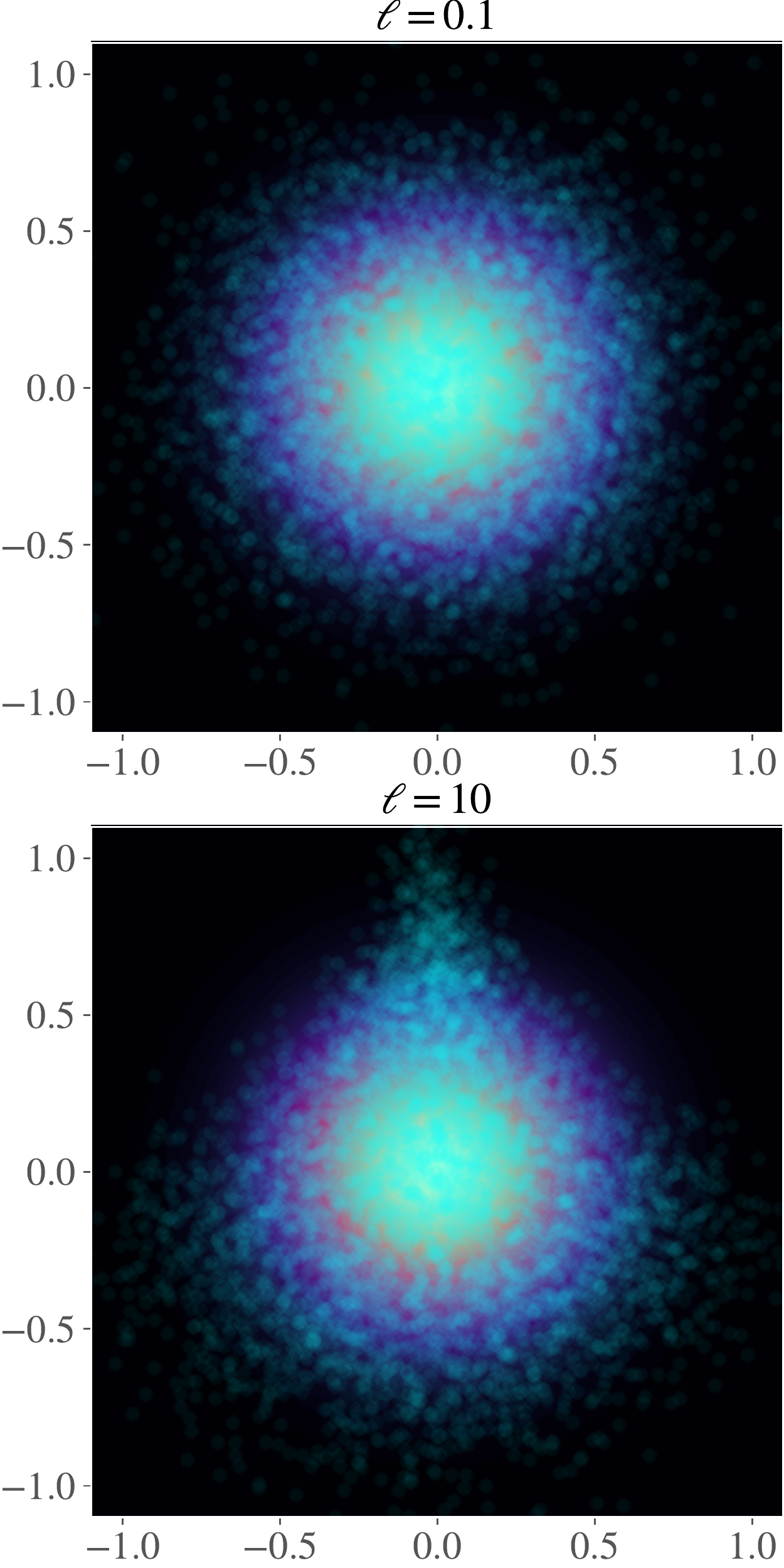}
        \caption{Gaussian $0.1I_2$}
    \end{subfigure} 
    \hfill
    \begin{subfigure}[b]{0.33\textwidth}
        \centering
        \includegraphics[width = \textwidth]{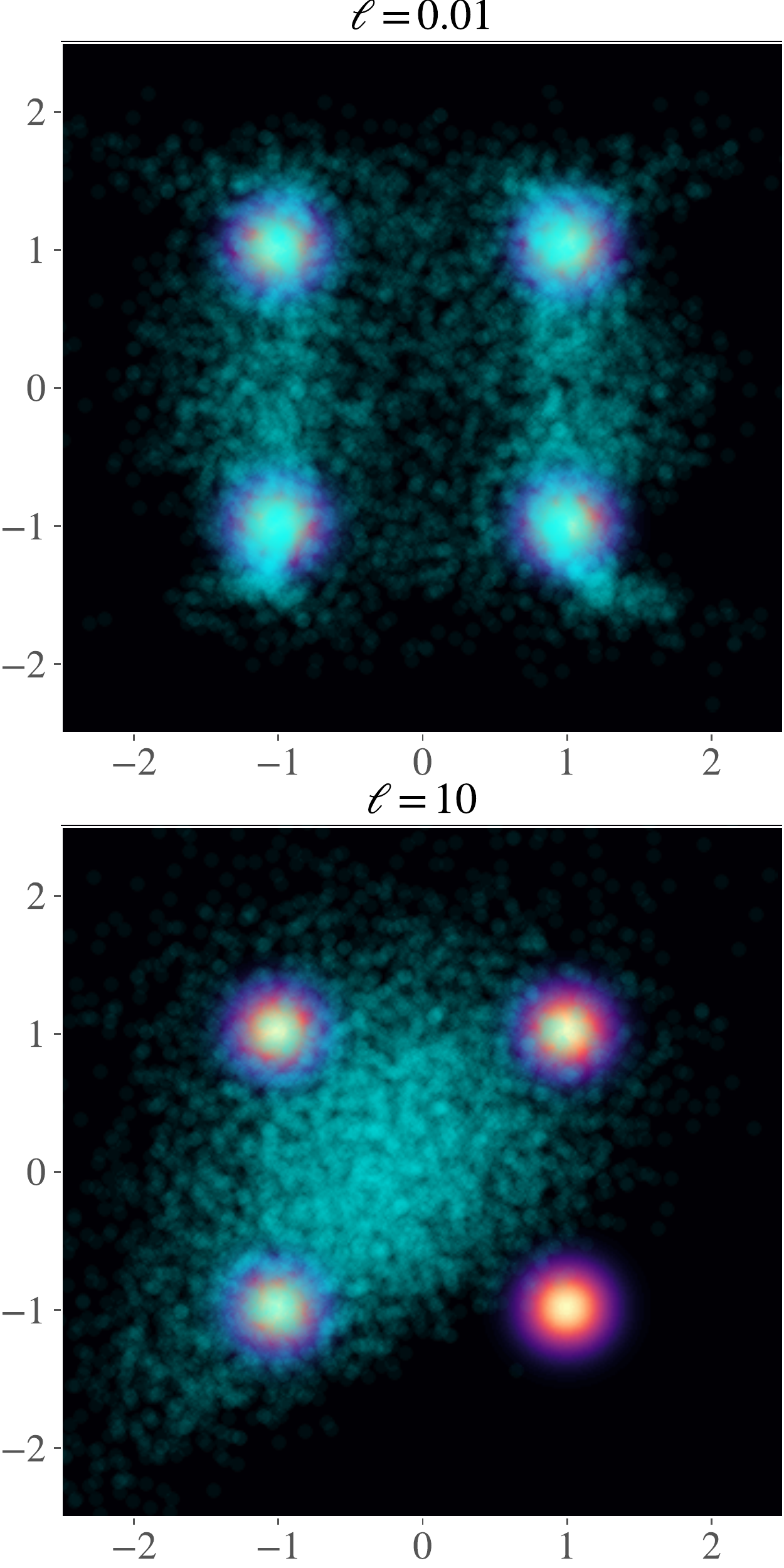}
        \caption{Multimodal}
    \end{subfigure} 
    \caption{KSD-based measure transport with varying choices of lengthscale $\ell$ for (a) a mean zero Gaussian target with covariance matrix $10I_2$, (b) a mean zero Gaussian target with covariance matrix $0.1I_2$ and (c) the multimodal problem encountered in \Cref{subsec: toy examples}. The lengthscale $\ell$ used in each experiment is indicated in the titles of the corresponding subplots.}
    \label{fig: effect of lengthscale}
\end{figure*}

\subsection{Investigating the effect of input dimension in the ReLU network transport map} \label{subsec: relu investigation}
In this section we investigate how changing input dimension of the ReLU network transport map can effect output. In order to isolate the input dimension as our variable of investigation, we fix the topology of the ReLU and consider ReLU networks of the form
\begin{equation*}
    f_{ReLU} =   F_3\circ \sigma \circ F_2 \circ \sigma \circ F_1,
\end{equation*}
where $F_1:\mathbb{R}^p\rightarrow\mathbb{R}^{20}$, $F_2:\mathbb{R}^{20}\rightarrow\mathbb{R}^{20}$ and $F_3:\mathbb{R}^{20}\rightarrow\mathbb{R}^2$ are affine transformations and $\sigma$ is the ReLU non-linearity, defined in \Cref{app: DNN}. We consider input dimensions $p = 1,2,3,4,5$. The target we considered is the multimodal target of \Cref{subsec: toy examples}. For each experiment, we pretrained each ReLU network on the reference distribution $\mathcal{N}((0,0),I_2)$ for $10,000$ iterations of Adam with learning rate $0.001$. We then trained on the multimodal target for $20,000$ iterations of Adam, again with learning rate $0.001$. Due to the random effects of initialisation (see \Cref{fig: multimodal failure}), we ran the experiments for $10$ different initialisations and report the best ones, see \Cref{table: relu input dimension}. The resulting transport maps are shown in \Cref{fig: relu input dimension output}. As we can see, the transport map struggles with the multiple modes when the input dimension was $p=1$ or $p=2$. 
\begin{table*}[h!] 
    \newcolumntype{Y}{>{\centering\arraybackslash}X}

\begin{tabularx}{\textwidth}{@{}lYYYYY@{}} \toprule
     Input dimension & {$p=1$} & {$p=2$} & {$p=3$} & {$p=4$} & $p=5$  \\ \midrule
     $W_1$ & $0.88$ & $0.39$ & $0.19$ & $0.27$ & $0.31$ \\ \bottomrule
\end{tabularx}

   \caption{The $W_1$ distance to the target $P$, as a function of the dimension $p$ of the reference distribution $Q$, using the ReLU neural network transport map. 
   }
\label{table: relu input dimension}
\end{table*}

\begin{figure}[h!] 
   \centering
   \includegraphics[width=1\textwidth]{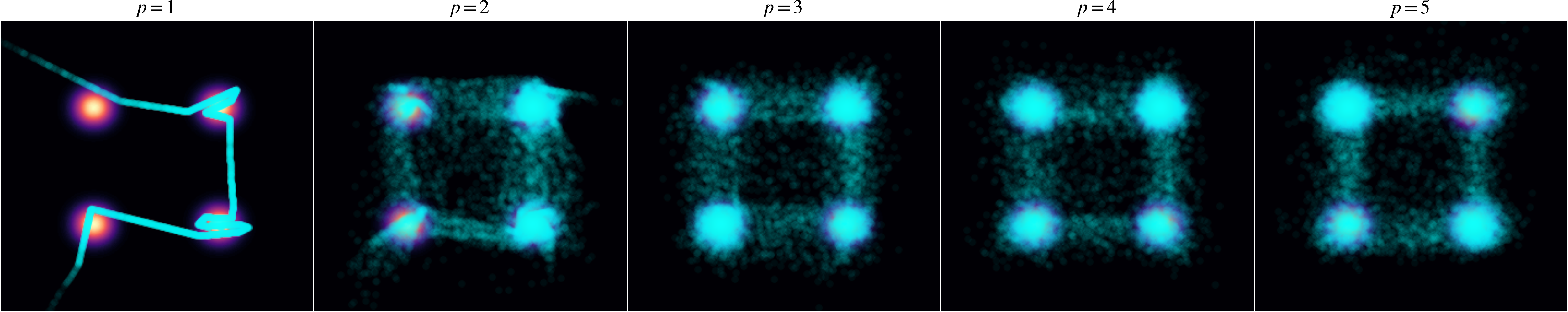}
\caption{
    Output of KSD-based measure transport using the ReLU network transport with varying input dimension. The input dimension used in each experiment is indicated in the titles of the corresponding subplots.
}
\label{fig: relu input dimension output}
\end{figure}

\subsection{Investigating the Effect of the U-statistic estimator vs. the V-statistic estimator} \label{subsec: U vs V}
Recall from \Cref{eq:KSD compute}, that the square of KSD is of the form
\begin{equation*}
    \mathcal{D}_S(P,P')^2 = \E_{Y,Y'\sim P'}\left[u_p(Y,Y')\right],
\end{equation*}
where $p$ is the density of $P$. For a given I.I.D. sample $\{y_i\}_{i=1}^n$ from $P'$, there are two natural estimators of $\mathcal{D}_S(P,P')^2$. The first is the \textit{V-statistic} (KSD-V),
\begin{equation*}
    \hat{\mathcal{D}}_S^V(P,P') = \frac{1}{n^2}\sum_{i=1}^n\sum_{j=1}^n u_p(y_i,y_j),
\end{equation*}
and the second is the \textit{U-statistic} (KSD-U),
\begin{equation*}
    \hat{\mathcal{D}}_S^U(P,P') = \frac{1}{n(n-1)}\sum_{1\leq i\neq j\leq n} u_p(y_i,y_j),
\end{equation*}
which is simply the V-statistic with the diagonal $i=j$ elements removed. The advantage of U-statistic is that it is unbiased and, for any given sample, provides the minimum-variance unbiased estimator (MVUE) \citep{liu2016kernelized}. On the other hand, the V-statistic provides a non-negative estimator, due to the positive-definiteness of $u_p$. 

To explore the differences between KSD-U and KSD-V as the objective, we re-ran the synthetic test bed experiments along with the majority of the transport maps in \Cref{subsec: toy examples}. Results are reported in \Cref{table: ksd u vs v}. It appears that, although KSD-U and KSD-V often have very close outcomes, KSD-U seems to be strictly better than KSD-V. The transport maps and their initialisation were the same as was used in \Cref{subsec: toy examples} and fully specified in \Cref{subsec: details of toy models}. For each experiment we used $10,000$ iterations of Adam with learning rate $0.001$.

\begin{table*}[t!] 
    \newcolumntype{Y}{>{\centering\arraybackslash}X}

\begin{tabularx}{\textwidth}{@{}llYYYYYY@{}} \toprule
        &  & \multicolumn{2}{c}{Sinusoidal}  & \multicolumn{2}{c}{Banana} & \multicolumn{2}{c}{Multimodal} \\
      \cmidrule(lr){3-4} \cmidrule(l){5-6} \cmidrule(l){7-8} 
     {Transport Map} & {$N$} & {KSD-U} & {KSD-V} & {KSD-U} & {KSD-V} &  {KSD-U} & {KSD-V} \\ \midrule
     IAF  & $10^4$                & $\textbf{0.38}$  & $0.39$ & $\bf{0.20}$ & 0.25  & 0.67 & $\bf{0.61}$ \\
     IAF (stable) & $10^4$        & $\textbf{0.35}$ & 0.36 & $\bf{0.16}$  &  0.19 & \textbf{0.61}  & $\bf{0.61}$ \\ 
     NAF   & $10^4$               & $\textbf{0.55}$ & 0.58  & $\bf{0.39}$  & 0.43 &  \textbf{0.095} & 0.12  \\ 
     SAF  & $10^4$                & $\textbf{0.23}$ & 0.27  & $\bf{0.20}$  & 0.48 & \textbf{0.30}  & 1.2 \\ 
     B-NAF  & $10^4$                & $\textbf{0.78}$ & 0.85  & $\bf{0.70}$  & $\bf{0.70}$ & $\bf{1.0}$   & $\bf{1.0}$  \\ 
     Polynomial (cubic)  & $10^4$ & $\textbf{0.40}$ & 0.61 & $\bf{0.25}$  & $0.41$ & 0.51  & $\bf{0.40}$ \\ 
    IAF mixture  & $3{\times}10^4$ & $\bf{1.3}$ & $\bf{1.3}$ & $\bf{0.19}$  & 0.39  & $\bf{0.037}$  & 0.040 \\ 
     ReLU network & $5{\times} 10^4$& $\textbf{0.71}$ & 0.96 & \textbf{0.43}  & 0.53  & \textbf{0.22}  &  1.2 \\\bottomrule
\end{tabularx}

   \caption{Results from the synthetic test-bed using either the U-statistic (KSD-U) or V-statistic (KSD-V) form of KSD as our objective. The first column indicates with parametric class of transport maps was used; full details for each class can be found in \Cref{subsec: details of toy models}. A map-dependent number of iterations of stochastic optimisation, $N$, are reported - this is to ensure all the optimisers approximately converged. The table reports the Wasserstein-1 metric between the approximation $T_\#Q$ and the target $P$. Bold values indicate which of KSD-U or KSD-V performed best for each transport map.
   }
\label{table: ksd u vs v}
\end{table*}

\subsection{Pathologies of KSD for Measure Transport} \label{subsec: pathologies}

Since KSD is a score-based method, it may exhibit similar pathologies to other score-based methods; see e.g. \cite{wenliang2020blindness}. In this section, we detail certain pathologies of KSD-based measure transport that we found experimentally and, if available, offer potential mitigation strategies.

\paragraph{Point Convergence:}{For small batch sizes (e.g. 25) in the sinusoidal synthetic experiment of \Cref{subsec: toy examples}, it was observed (albeit rarely) that, when using the ReLU transport map, the transport map converged to a limit in which all inputs were mapped to the origin. 
The origin is the mode of the sinusoidal synthetic density. This only occurred when using a degenerate initialisation and for small batch sizes.
Due to the tightness of the sinusoidal target and the resulting large score values at points $x$ diverging away slightly from the support of the target, the KSD value of output for all the transport maps used in \Cref{subsec: toy examples} was generally of the order $10^8$. However, for the approximating transport map maps everything to the origin, the resulting KSD score was $200$ with $\ell = 0.1$. 
Thus, from the perspective of KSD, a transport map that maps everything to the origin and thus having poor Wasserstein distance, was considered better than transport maps that obtain smaller Wasserstein distances. 

This problem was mitigated by using better initialisations, larger batch sizes and pretraining on the reference distribution. 
}

\paragraph{Multimodal Failure:}{It was found that, particularly with ReLU neural network transport map, the inferred transport map for the multimodal synthetic distribution in \Cref{subsec: toy examples} could fail to find all four high density regions in the target. 
The outcome was highly dependent on the initialisation and even persisted when pretraining on the reference distribution.  For example, in \Cref{fig: multimodal failure} we plot three random initialisations of the ReLU transport map used in \Cref{subsec: toy examples}. Each transport map was pretrained on their reference distribution for $10,000$ iterations of Adam. The KSD estimates with $\ell = 0.1$ using $10,000$ samples from the approximating distributions were $0.194$, $0.205$ and $0.147$ from left to right respectively. The three outputs achieved broadly similar KSD scores, indicating that KSD was not able to differentiate between these outcomes. This is problem is further exacerbated as the number of modes of the target increases.
For discussion of remedies, see \cite{wenliang2020blindness}.

\begin{figure}[h!] 
   \centering
   \includegraphics[width=1\textwidth]{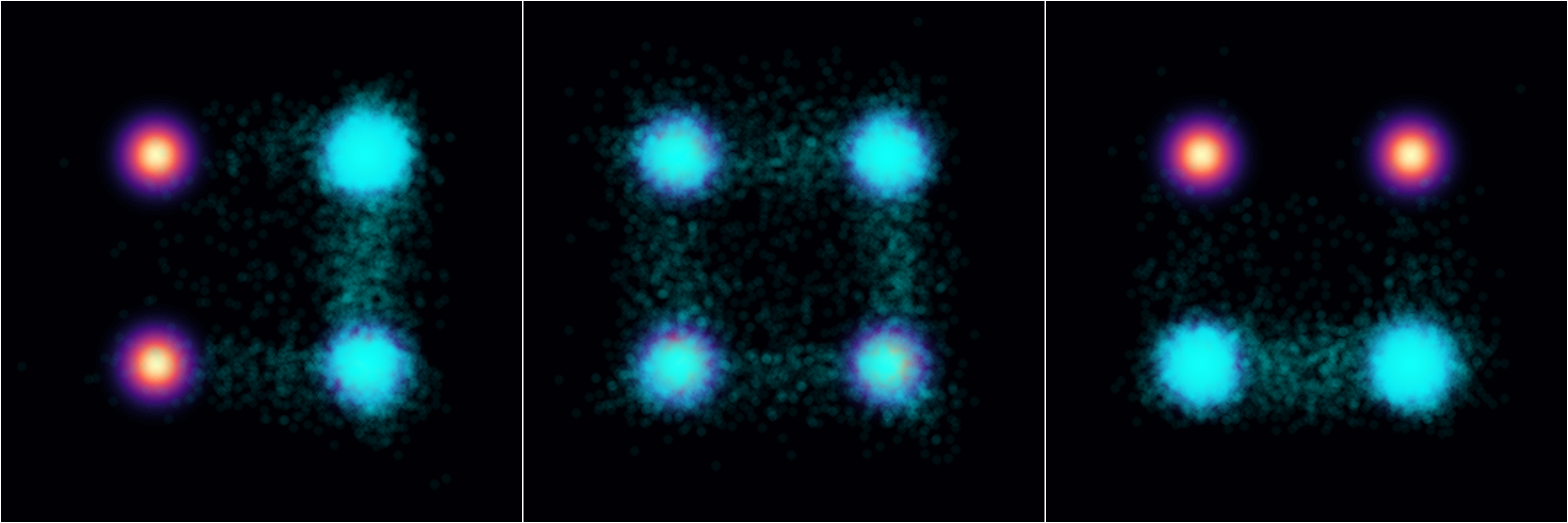}
\caption{
    Output of KSD based measure transport with three random initialisations of the ReLU transport map. Each transport map was pretrained on their reference distribution for $10,000$ iterations of Adam. The KSD estimates with $\ell = 0.1$ using $10,000$ samples from the approximating distributions were $0.194$, $0.205$ and $0.147$ from left to right respectively.
}
\label{fig: multimodal failure}
\end{figure}
}

\paragraph{Poor Choice of Lengthscale:}
Finally, as we have seen in \Cref{subsec: misspecified lengthscale}, if the choice of $\ell$ is poor, the resulting output of KSD based measure transport can exhibit pathologies. This is demonstrated, for instance, in \Cref{fig: effect of lengthscale}.
This issue is remedied by, for example, employing the median heuristic \cite{garreau_2018} to set the length-scale parameter in the kernel.


\end{document}